\documentclass[letterpaper,11pt]{article}

\usepackage{geometry}
\usepackage{graphicx}
\usepackage{amsmath}
\usepackage{amsthm}
\usepackage{amsfonts}
\usepackage{amssymb}
\usepackage{url}

\usepackage[font=sf]{caption}

\begin{document}

\title{{\bf A general theory of equilibrium behavior}}

\author{{\large{ Ioannis Avramopoulos\footnote{The author is with the National Technical University of Athens. His email is \texttt{ioannis.avramopoulos@gmail.com.}}}}
}

\date{}

\maketitle

\thispagestyle{empty} 

\newtheorem{definition}{Definition}
\newtheorem{proposition}{Proposition}
\newtheorem{theorem}{Theorem}
\newtheorem{corollary}{Corollary}
\newtheorem{lemma}{Lemma}
\newtheorem{axiom}{Axiom}
\newtheorem{thesis}{Thesis}

\vspace*{-0.2truecm}

\begin{abstract}
Economists were content with the concept of the Nash equilibrium as game theory's {\em solution concept} until Daskalakis, Goldberg, and Papadimitriou showed that finding a Nash equilibrium is most likely a computationally hard problem, a result that set off a deep scientific crisis. Motivated, in part, by their result, in this paper, we propose a general theory of equilibrium behavior in {\em vector fields} (and, therefore, also noncooperative games). Our line of discourse is to show that these universal in nature mathematical objects are endowed with significant structure, which we probe to unearth atypical, previously unidentified, equilibrium behavior.
\end{abstract}

\section{Introduction}

There is evidence in the literature that various independently founded disciplines (including game theory, nonlinear optimization theory, and dynamical systems theory) have in fact been founded on variations of the same concept with the formalization of the {\em variational inequality problem} being a case in point. In the finite-dimensional variational inequality problem (for example, see~\cite{Bertsekas-Tsitsiklis}), we are given a set $X \subseteq \mathbb{R}^m$, which is typically assumed to be nonempty, closed, and convex, and a vector field $c: \mathbb{R}^m \rightarrow \mathbb{R}^m$, and our objective is to find a vector $x^* \in X$ such that
\begin{align*}
(x - x^*) \cdot c(x^*) \geq 0, \forall x \in X.
\end{align*}
We will refer to the solutions of the variational inequality problem as the {\em critical elements} of the vector field $(X,c)$. The concept of critical elements has been shown to coincide with the concept of Nash equilibria (for example, see~\cite{PopulationGames}), the concept of critical points of nonlinear optimization problems (for example, see~\cite{Bertsekas-Tsitsiklis}) and the concept of equilibria of dynamical systems (for example, see~\cite{Dupuis}). In this paper, we make a case toward unifying the foundations of the aforementioned disciplines by proposing a {\em solution concept} for a general vector field. The term {\em solution concept} is to be understood in the sense of game theory, namely, a prediction on the behavior of a system. First we argue that critical elements do not qualify as a solution concept for a general vector field.

\subsection{Critical elements are not a solution concept}

Suppose momentarily that our effort to lay a foundation of a general theory of vector fields whose objective is to make predictions through its solution concept about the behavior of systems that evolve according to vector fields is a rightful pursuit. Then we claim that critical elements cannot be that solution concept. For if critical elements qualified as such a solution concept, they would certainly need to qualify as the solutions of vector fields that have special structure, such as, for example, {\em gradient fields}, that is, vector fields that are generated by taking the gradient of a scalar function. The problem of devising a solution concept for gradient fields has been the subject of {\em nonlinear optimization theory} where it was recognized early on (in fact, by Fermat himself who formulated optimization's first solution concepts) that critical elements are only a {\em necessary optimality condition} and by no means sufficient. (For example, consider the problem of optimizing $f(x) = x^3$. In this problem, $x=0$ is a critical point, that is, a solution of the aforementioned variational inequality, whereas this point does not satisfy any intuitive notion of optimality.) 

In the same vein, since noncooperative games can also be represented as vector fields (for example, see~\cite{PopulationGames}), does there exist solid ground to justify accepting the concept of Nash equilibrium as game theory's solution concept? Our thesis is that it doesn't; for one thing, computing Nash equilibria is most likely a computationally hard problem~\cite{Daskalakis}. What's more, the intersection of nonlinear optimization problems and noncooperative games is not empty: Many interesting noncooperative games are gradient fields in that they admit a {\em potential function} (for example, see~\cite{PotentialGames}). Since the potential function fully characterizes the incentive structure of games that have one, on what grounds should the critical points qualify as a solution concept for these games given that critical points do not qualify to be a solution concept in the general nonlinear optimization problem?

\subsection{If not critical elements, what then?}

Our thesis in this paper is that the pursuit of a general theory of vector fields is rightful, and to support this thesis we propose a {\em general theory of equilibrium behavior}. The formalization of such a theory may sound at first to be a formidable and copious task considering that the concept of vector fields is as universal in mathematics as the Turing machine is in computer science, if not more. We believe, however, that our line of discourse is easy to follow as, in fact, the foundation of the theory is based on an elementary and rather intuitive idea: {\em The elements of a vector field are comparable, and, therefore, can be ordered, in a meaningful manner}. Once this idea is in place the question of devising an appropriate means by which to order them becomes a simpler (although perhaps not straightforward) task. The induced order is not, to the extent of our knowledge, a mathematical structure known prior to this work, and we take the privilege of naming it a {\em polyorder}.

\subsection{The main ideas in our theory}

Consider a gradient field defined over a closed and convex set $X \subseteq \mathbb{R}^m$ and let $f$ be its scalar potential function. Let $x,y \in X$ be arbitrary elements of this gradient field, and consider a (topological) path $p$ in $X$ starting at $x$ and ending at $y$ whose image is the convex hull of $x$ and $y$ (that is, all convex combinations of $x$ and $y$). Call $p$ a {\em descent path} if the following conditions are met: (1) $f(x) > f(y)$ and (2) $f$ is everywhere nonascending along $p$. If these conditions are met {\em order} $x$ and $y$ such that $y$ is ``better than'' $x$. Similarly, call $p$ a {\em ascent path} if $f(x) < f(y)$ and if $f$ is everywhere nondescending along $p$, and, if so, {\em order} $x$ and $y$ such that $x$ is ``better than'' $y$. A polyorder then, in this particular example of a gradient field, is the induced mathematical structure obtained by the pairwise comparison of all elements in the aforementioned way.

Now, let $x^*_{\min}$ be an element of $X$ such that for all other elements $x$ of $X$ the corresponding path from $x^*_{\min}$ to $x$ as defined above is {\em not} a descent path, and call such an element {\em minimal}; minimal elements are precisely those elements that are ``no worse'' than all other elements. Also, let $x^*_{\max}$ be an element of $X$ such that for all other elements $x$ of $X$ the corresponding path from $x^*_{\max}$ to $x$ is {\em not} an ascent path, and call such an element {\em maximal}; maximal elements are precisely those elements that are ``no better'' than all other elements. 
 
 Every polyorder has a {\em minimum solution concept} and a {\em maximum solution concept}, the former being the set of all minimal elements of $X$ and the latter being the set of all maximal elements of $X$. These ideas generalize to general vector fields, and such generalized orders form the basis of our theory of equilibrium behavior. One of the main contributions of this paper is to show how polyorders encompass various known concepts in nonlinear optimization theory and game theory and how they can evince interesting yet unidentified prior to this work behavior.\\

\noindent
The rest of this paper is organized as follows: We motivate and formalize our theory of equilibrium behavior in the next section, while in Section~\ref{qowieurksdfjhgskdjgh} we explore the precise relationship between our theory and the theories of nonlinear optimization and games. Thereafter, in Section~\ref{aslkdjfhsdkjfhsdjf}, we study atypical previously unknown equilibrium behavior that our theory predicts. Finally, in Section~\ref{qwoeriwueroweiru}, we discuss related work.

\section{Polytropic optimization: In search of equilibrium}
\label{alsskdjfhskdjfhsdkjhf}

In this section, we formalize our theory. We start by motivating the ideas in this theory using ideas from {\em nonlinear optimization theory} and {\em evolutionary game theory}, then we introduce some basic concepts from the {\em theory of preference relations}, and, finally, we present a precise formalization of what we may call a {\em general theory of equilibrium behavior in vector fields}.


\if(0)

In this section, we present a {\em vector field optimization problem}. Our use of the term ``optimization'' is in a manner analogous to that of classical optimization theory (a theory that is concerned with the study of {\em scalar fields}), which may seem unintuitive at first since, 

Classical optimization theory is concerned with the problem of ``minimizing'' scalar fields of the form $f : X \subseteq \mathbb{R}^m \rightarrow \mathbb{R}$, which implies that to every element $x \in X$ corresponds a scalar value $f(x)$. However, given two elements $x,y \in X$, would $f(x) < f(y)$ imply that $x$ is ``better'' than $y$? Perhaps in some applications it would, however, in general, it doesn't, and if it is the latter, the question of devising a natural order of vector field elements remains even in scalar fields.

-----

and observe that the natural order of the elements of this field by their value is not always meaningful. Therefore the question of a natural order of the elements of fields remains even for scalar ones. In this section, we present such a natural order, natural in the sense that it induces a structure which contains many well known concepts and reveals yet unknown ones. 

\fi

\subsection{Some motivating ideas in nonlinear programming}

We start off with some basic ideas from {\em nonlinear programming}, the algorithmic component of nonlinear optimization theory. The archetypical problem of nonlinear optimization theory is to ``minimize'' a scalar {\em objective function} $f$ subject to a {\em constraint set} $X \subseteq \mathbb{R}^m$. In standard expositions of the theory, nonlinear programming algorithms would ideally compute {\em global optima}, an objective, however, that is out of the scope of existing algorithms. This observation motivates the study of algorithms with local convergence properties (for example, see~\cite{Bertsekas}) and the study of optimization problems where the concepts of local and global minima coincide (for example, see~\cite{ConvexOptimization, ConvexAnalysis}).

Careful examination of the line of discourse of nonlinear optimization theory reveals that the concept of the objective function is {\em not} a first-order concept of this theory. Indeed, if the scope of nonlinear programming algorithms were to compute global optima, it would have been, however, as a point of fact, the objective function is only relevant inasmuch as it permits the convenient definition of {\em descent directions}. 

Given a point $x$ in the constraint set $X$, a descent direction for this point is a vector in the corresponding tangent space of the constraint set such that the objective function $f$ decreases locally. What nonlinear programming algorithms essentially do is to {\em follow descent directions in carefully selected steps and stop at points lacking descent directions.} The concept of a decent direction is, therefore, arguably more important than that of the objective function.

Most texts on nonlinear optimization theory motivate the concepts of {\em local minima} and {\em strict local minima} as forming the basis of this theory's ``solution concept'' in the sense that once nonlinear programming algorithms (in their course of following descent directions) arrive at such points, they stop. In this perspective, this paper can be understood as the following thought experiment: What precisely are the points that lack descent directions and do the aforementioned minimality concepts suffice to characterize these points?  The precise definition of a descent direction will turn out to be crucial in this thought experiment.

\subsection{Some motivating ideas in evolutionary game theory}

We continue with some basic ideas from {\em evolutionary game theory} whose fundamental object of study is the {\em population game}, which is typically understood as a mathematical model of the strategic interaction among a large number of anonymous infinitesimal agents. Mathematically population games are {\em vector fields} constrained on a polyhedron that  in fact generalize normal-form games.

\subsubsection{Population games}

A population game is a pair  $(X, c)$. $X$, the game's {\em state space} or {\em strategy profile space}, has product form, i.e., $X = X_1 \times \cdots \times X_n$, where the $X_i$'s are simplexes and $i=1,\ldots,n$ refers to a {\em player position} (or {\em population}). To each player position corresponds a set $S_i$ of $m_i$ {\em pure strategies} and a {\em mass} $\omega_i > 0$ (the population mass). The {\em strategy space} $X_i$ of player position $i$ has the form
\begin{align*}
X_i = \left\{ x \in \mathbb{R}^{m_i}  \bigg| \sum_{j \in S_i} x_j  = \omega_i \right\}.
\end{align*}
We refer to the elements of $X$ as {\em states} or {\em strategy profiles}. Each strategy profile $x \in X$ can be decomposed into a vector strategies, i.e., $x = (x_1,\ldots,x_n)'$. Let $m = \sum_i m_i$. $c: X \rightarrow \mathbb{R}^m$, the game's {\em cost function}, maps $X$, the game's state space, to vectors of {\em costs} where each position in the vector corresponds to a pure strategy. It is typically assumed that $c$ is continuous.

\subsubsection{Solution concepts for population games}

Even before the result of Daskalakis et al. disputed the Nash equilibrium as game theory's solution concept, the theory of population games was based on equilibrium refinements such as the {\em evolutionarily stable strategy}~\cite{TheLogicOfAnimalConflict, Evolution} as well as variants and generalizations of this concept (for example, see~\cite{Thomas}). In a manner analogous to nonlinear optimization theory, most texts on evolutionary game theory motivate the aforementioned equilibrium refinements as being this theory's ``solution concept'' in the sense that they are points (or sets thereof) where {\em evolution stops}. 

In contrast to nonlinear optimization theory where algorithmic ideas abound, evolutionary game theory is concerned perhaps exclusively with the concept of {\em dynamics} that although arguably would qualify as an inherently algorithmic concept from the perspective of computer scientists, it formally views evolution as a {\em continuous-time process} (with few exceptions). Evolutionarily stable strategies as well as their variants and generalizations have been shown to be attractive under a wide range of dynamics (including perhaps most prominently among them the {\em replicator dynamic}~\cite{TaylorJonker}). The precise definition of an evolutionarily stable strategy is based on an idea whose significance in our theory of equilibrium behavior cannot be understated.

\subsubsection{The concept of invasion}

One of the fundamental concepts in evolutionary game theory is that of {\em invasion}. Let $(X, c)$ be a population game, and let $x, y \in X$. We say that $x$ invades $y$ if $x \cdot c(y) < y \cdot c(y)$. The standard evolutionary interpretation of invasion is to consider a population of organisms that form a continuum mass, to interpret $y$ as the {\em state} of this population, and $y \cdot c(y)$ as the {\em survival cost} of an organism being in that state. In this vein, $x \cdot c(y)$ is the survival cost of a mutant organism whose state switches from $y$ to $x$. Then the condition that $x$ invades $y$ implies that the population of mutants will proliferate in the incumbent population due to the favorable survival cost.

There is an alternative interpretation of the concept of invasion, one that is more in alignment with ideas in nonlinear optimization theory. To that end, let's write the condition that $x$ invades $y$ as $(x-y) \cdot c(y) < 0$, and note that this implies that the angle between $c(y)$ and $(x-y)$ is obtuse. Stated differently, the projection of $c(y)$ on the line segment connecting $x$ and $y$ ``points away'' from $x$. Suppose now that the vector field is the gradient field of a scalar potential. Then the condition that $x$ invades $y$ is precisely the condition that $(x-y)$ is a {\em descent direction} of the scalar field at $y$ (a simple fact which is easy to prove by a Taylor expansion). 

What if the vector field is not a gradient field though? Are we then justified to think of $(x-y)$ as a ``descent direction'' in {\em some} sense? Our intuitive understanding of a descent direction stipulates that along such a direction the objective function must decrease locally, however, as discussed earlier in the setting of nonlinear optimization theory, the objective function is {\em not} a first order concept, and, therefore, the possibility of witnessing some other more elementary phenomenon than the descent of the potential function (one that transcends to general vector fields) is open. With the benefit of hindsight, we can indeed assert that what is happening at a more elementary level is what we may call the descent of an abstract notion of an ``order,'' one we yet need to define. To that end, we turn to some foundational ideas from the theory of preference relations.


\if(0)

\subsubsection{Game theory}

Myerson defines Game Theory as ``[T]he study of mathematical models of conflict and cooperation between intelligent rational decision-makers''~\cite{Myerson}. Game theory's most popular object of study is the {\em normal-form game}, which is a triple $$(I, (S_i)_{i \in I}, (u_i)_{i \in I}),$$ where $I$ is the set of players, $S_i$ is the set of pure strategies available to player $i$, and $u_i: S \rightarrow \mathbb{R}$ is the utility function of player $i$ where $S = \times_i S_i$ is the set of all strategy profiles (combinations of strategies). Normal-form games are {\em particular instances} of a more general concept, that of {\em population games} (for example, see~\cite{PopulationGames}), which are discussed below where the connection between normal-form games and vector fields becomes immediately apparent (an observation that has perhaps not received appropriate attention in the community).

\subsubsection{Evolutionary game theory}

The origins of evolutionary game theory are traced in biology~\cite{xxx, yyy} in contrast with game theory whose origins are traced in the study of economic behavior~\cite{xxx}. However, the mathematical formalisms of these disciplines are in fact quite closely related to each other, and our primary focus in this paper is on abstractions.

In this paper, we are concerned with {\em population games} of which normal form games are a special case (for example, see~\cite{PopulationGames}). A population game $\mathcal{G}$ is a pair  $(X, c)$. $X$, the game's {\em state space} or {\em strategy profile space}, has product form, i.e., $X = X_1 \times \cdots \times X_n$, where the $X_i$'s are simplexes and $i=1,\ldots,n$ refers to a {\em player position} (or {\em population}). To each player position corresponds a set $S_i$ of $m_i$ {\em pure strategies} and a {\em mass} $\omega_i > 0$ (the population mass). The {\em strategy space} $X_i$ of player position $i$ has the form
\begin{align*}
X_i = \left\{ x \in \mathbb{R}^{m_i}  \bigg| \sum_{j \in S_i} x_j  = \omega_i \right\}.
\end{align*}
We refer to the elements of $X$ as {\em states} or {\em strategy profiles}. Each strategy profile $x \in X$ can be decomposed into a vector strategies, i.e., $x = (x_1,\ldots,x_n)'$. Let $m = \sum_i m_i$. $c: X \rightarrow \mathbb{R}^m$, the game's {\em cost function}, maps $X$, the game's state space, to vectors of {\em costs} where each position in the vector corresponds to a pure strategy. We assume that $c$ is {\em continuous}.

Population games are the first fundamental object of study of {\em evolutionary game theory}, noting that the origins of this theory are traced in biology~\cite{xxx} in contrast with game theory whose origins are traced in the study of economic behavior~\cite{xxx}. However, the mathematical formalisms of these disciplines are in fact quite closely related to each other as noted above, and our primary focus in this paper is on abstractions. In this vein, we will abstract population games (and, therefore, also normal-form games) as {\em vector fields} (that is, maps of the form $F: X \subseteq \mathbb{R}^m \rightarrow \mathbb{R}^m$). 

The second fundamental object of study of evolutionary game theory is the {\em evolutionary dynamics} of {\em revision protocols}~\cite{PopulationGames}. Revision protocols are simple procedures that population members employ in their effort to ``navigate'' a population game. (This is again in contrast with game theory whose assumption is that the ability to navigate a game is derived from introspection and rationality). Revision protocols are studied as evolution rules of {\em dynamical systems}. 

\fi

\subsection{Some ideas in preference relations theory}

Let $X$ be a nonempty set. A subset $R$ of $X \times X$ is called a {\em binary relation} on X. If $(x,y) \in R$, we write $xRy$, and if $(x,y) \not\in R$, we write $\neg xRy$. $R$ is called {\em reflexive} if $xRx$ for each $x \in X$ and {\em complete} if for each $x,y \in X$ either $xRy$ or $yRx$. It is called {\em symmetric} if, for any $x, y \in X$, we have that $xRy \Rightarrow yRx$, and {\em antisymmetric} if, for any $x, y \in X$, $xRy \wedge yRx \Rightarrow x=y$. Finally, $R$ is called {\em transitive} if for any $x, y, z \in X$ we have that $xRy$ and $yRz$ imply $xRz$.

Let $xPy \Leftrightarrow xRy \wedge \neg yRx$ and $xIy \Leftrightarrow xRy \wedge yRx$. Then $P$ and $I$ are also binary relations on X where $P \subseteq R$ and $I \subseteq R$. $P$ is called the {\em asymmetric} (or {\em strict}) part of $R$ and $I$ is called its symmetric part. If $P$ is antisymmetric, then $R$ is called a {\em preference relation}.

Preference relations are a fundamental object of study of {\em microeconomic theory}. The most elementary problem in this theory is that of {\em individual decision making} (see, for example,~\cite{MicroeconomicTheory}) whose starting point is a set of alternatives (say $X$) among which an individual must choose. A fundamental tenet of microeconomic theory (since the definitive work of Von Neumann and Morgernstern~\cite{TheoryOfGames}) is that individual preferences are {\em rational}.

\begin{definition}
The preference relation $R$ is {\em rational} if it is {\em complete} and {\em transitive}. 
\end{definition}

The study of transitive preference relations is the subject of {\em order theory} (for example, see~\cite{Ok}); although this theory is of considerable intellectual merit, in the rest of this paper, we will be concerned with preference relations that are generally non-transitive.

If a preference relation is not transitive, one would hope that it is perhaps acyclic.

\begin{definition}
The preference relation $R$ is {\em acyclic} if there is no finite set $\{ x_1, \ldots, x_n \} \subset X$ such that $x_i R x_{i+1}$ for $i = 1,\ldots,n-1$ and such that $x_n R x_1$. 
\end{definition}

There is a growing theory of acyclic preference relations (for example, see~\cite{Bergstrom, Walker, Alcantud, Salonen-Vartiainen}), however, the preference relations we will be concerned with are in general not even acyclic.

\begin{definition}
Let $R$ be a preference relation on $X$ and let $x, y \in X$. If $xPy$, where $P$ is the asymmetric part of $R$, we say that $x$ {\em dominates} $y$. We further say that an element of $X$ is {\em minimal} (or {\em undominated}) if there is no element in $X$ that dominates it. We, finally, say that an element of $X$ is {\em maximal} if there no element in $X$ that it dominates.
\end{definition}

\subsection{Polyorders and polytropic optimization}

In the rest of this section, we introduce our solution concepts for scalar and vector fields.

\subsubsection{Preliminaries}

Before defining our solution concepts, however, let's ponder momentarily why critical elements do not qualify as such. It seems intuitive to us that, in any sensible notion of optimality, points where it is possible to descend from locally should not qualify as being optimal, and it is precisely this notion of optimality that critical elements fail to pass. Consider, for example, the problem of minimizing a scalar field and an interior strict local maximum of this field. It is easy to see that the (interior) strict local maximizer is a critical element as its gradient vanishes. However, starting from the maximizer and following any direction in the tangent space, the objective function {\em clearly} does decrease, and, therefore, the maximizer cannot qualify as a solution to our minimization problem.

Our definition of an appropriate solution concept will be based on the machinery of the theory of preference relations, and, in particular, the concept of {\em minimal} or {\em undominated} elements. You may observe in the previous definition of minimality that it is a concept defined by a property that it {\em lacks}, and since we are interested in mathematically capturing the idea of points that lack descent directions, this confluence is particularly suiting. However, devising an appropriate preference ordering is not a straightforward task. To see this, recall the definition of critical elements: Given a vector field $c : X \rightarrow \mathbb{R}^m$, $x^* \in X$ is a critical element if
\begin{align*}
(x - x^*) \cdot c(x^*) \geq 0, \forall x \in X.
\end{align*}
Notice in the definition that critical elements are also defined by what they lack: $x^*$ is a critical element precisely when it is {\em uninvadable}. Our solution concept requires a new, but straightforward in its formalism, mathematical concept, which we are going to call {\em weak dominance} in whose definition we will need the following elementary definition from topology.

\begin{definition}
A {\em path} $p$ in $X \subseteq \mathbb{R}^m$ is a continuous map $p : [0,1] \rightarrow X$.
\end{definition}

\subsubsection{Linear scalar polyorders}

Let $f: X \subseteq \mathbb{R}^m \rightarrow \mathbb{R}$ be a continuous scalar field where $X$ is closed and convex. Let $x, y \in X$, and define a path $y_\epsilon : [0,1] \rightarrow X$ such that $y_\epsilon = \epsilon x + (1-\epsilon) y$. Let $\preceq_f$ be a preference relation on $X$ such that 
\begin{align*}
x \preceq_f y \Leftrightarrow \forall \epsilon_1, \epsilon_2 \in [0,1]: \epsilon_1 < \epsilon_2 \Rightarrow f(y_{\epsilon_1}) \geq f(y_{\epsilon_2}).
\end{align*} 
If $x \preceq_f y$, we say that $y_\epsilon$ is a path of {\em weak ascent} from $x$ to $y$ (and a path of {\em weak descent} from $y$ to $x$). We call $(X, \preceq_f)$ a {\em linear scalar polyorder}. (Note that $y_{\epsilon_1}$ is ``closer'' to $y$ than $y_{\epsilon_2}$.)

\begin{proposition}
If $x \prec_f y$, then $f(x) < f(y)$.
\end{proposition}

\begin{proof}
Since $x \prec_f y$, we have that $x \preceq_f y$ and $\neg(y \preceq_f x)$. Since $x \preceq_f y$, we have that $\forall \epsilon_1, \epsilon_2 \in [0,1]: \epsilon_1 < \epsilon_2 \Rightarrow f(y_{\epsilon_1}) \geq f(y_{\epsilon_2})$. Furthermore, since $\neg(y \preceq_f x)$, there exist $\epsilon_1, \epsilon_2 \in [0,1]$, where $\epsilon_1 < \epsilon_2$ such that $f(y_{\epsilon_1}) > f(y_{\epsilon_2})$. Since, $f(x) \leq f(y_{\epsilon_2})$ and $f(y_{\epsilon_1}) \leq f(y)$, the proof is complete.
\end{proof}

By virtue of the previous proposition, if $x \prec_f y$, we say that $y_\epsilon$ is an {\em ascent path} from $x$ to $y$ (and that $y-x$ is an {\em ascent direction} at $x$) and a {\em descent path} from $y$ to $x$ (and that $x-y$ is a {\em descent direction} at $y$). However, note that an ascent or descent direction may not be {\em locally} so in the sense that in a neighborhood of the corresponding point it may well be the case that $f$ is constant.

\begin{proposition}
\label{wpoeritueoiriut}
The strict part of a linear scalar polyorder is acyclic.
\end{proposition}

\begin{proof}
Suppose there exists a finite set $\{ x_1, \ldots, x_n \} \subset X$ such that $x_i \prec_f x_{i+1}$ for $i = 1,\ldots,n-1$ and such that $x_n \prec_f x_1$. Then $f(x_i) < f(x_{i+1})$ for $i = 1,\ldots,n-1$ and $f(x_n) < f(x_1)$, which contradicts the assumption that $f$ is single-valued.
\end{proof}

The previous concepts naturally generalize to general vector fields as shown next.

\subsubsection{Linear vector polyorders}

Let $c: X \subseteq \mathbb{R}^m \rightarrow \mathbb{R}^m$ be a continuous vector field where $X$ is closed and convex. Let $x, y \in X$, and define a path $y_\epsilon : [0,1] \rightarrow X$ such that $y_\epsilon = \epsilon x + (1-\epsilon) y$. Let $\preceq_c$ be a preference relation on $X$ such that 
\begin{align*}
x \preceq_c y \Leftrightarrow \forall \epsilon \in [0,1]: x \cdot c(y_\epsilon) \leq y \cdot c(y_\epsilon).
\end{align*} 
If $x \preceq_c y$, we say that $x$ {\em weakly dominates} $y$. We call $(X, \preceq_c)$ a {\em linear vector polyorder}. We may understand weak dominance as follows: If $x$ weakly dominates $y$, for all $\epsilon \in [0,1]$, the projection of $c(y_\epsilon)$ on the path $y_\epsilon$ does not point toward $x$.

\begin{proposition}
If $x \prec_c y$, then $\forall \epsilon \in [0,1]: x \cdot c(y_\epsilon) \leq y \cdot c(y_\epsilon)$ and $\exists \epsilon \in [0,1]: x \cdot c(y_\epsilon) < y \cdot c(y_\epsilon)$.
\end{proposition}

\begin{proposition}
The strict part of a linear vector polyorder is antisymmetric.
\end{proposition}

\begin{proof}
The proof easily follows from the previous proposition.
\end{proof}

We now show that the concept of `weak dominance' generalizes the concept of `weak descent'. 

\begin{proposition}
Let $\tilde{c}: \mathbb{R}^m \rightarrow \mathbb{R}^m$ be a continuous vector field and let $c: X \subseteq \mathbb{R}^m \rightarrow \mathbb{R}^m$ be its restriction to the closed and convex set $X$. Suppose that $\tilde{c}$ is a gradient field and let $f$ be its potential function. Then $x \preceq_c y \Leftrightarrow x \preceq_f y$.
\end{proposition}

\begin{proof}
We show that $x \preceq_c y \Rightarrow x \preceq_f y$, noting that the reverse direction is completely analogous. To that end, let $x,y \in X$ such that $x \preceq_c y$, and suppose there exist $\epsilon_1, \epsilon_2 \in [0,1]$ where $\epsilon_1 < \epsilon_2$ such that $f(y_{\epsilon_1}) < f(y_{\epsilon_2})$. Then by the mean value theorem there exists $\epsilon \in (\epsilon_1, \epsilon_2)$ such that
\begin{align*}
 \nabla f(y_\epsilon) \cdot (x-y) = \frac{1}{\epsilon_2-\epsilon_1} \nabla f(y_\epsilon) \cdot (y_{\epsilon_2} - y_{\epsilon_1}) = f(y_{\epsilon_2}) - f(y_{\epsilon_1}) > 0.
\end{align*} 
This contradicts our assumption that $x \preceq_c y$.
\end{proof}

\subsubsection{Solution concepts for scalar and vector fields}

What we have so far accomplished is to introduce natural, we believe, orderings of general scalar and vector fields. It is now a small step to define our solution concepts based on these orderings, definitions which we state as the following theses.

\begin{thesis}[{\bf First fundamental thesis of polytropic optimization}]
The minimum (resp. maximum) solution concept of a continuous scalar field that is defined over a closed and convex domain is the set of minimal (resp. maximal) elements of its linear scalar polyorder.
\end{thesis}

\begin{thesis}[{\bf Second fundamental thesis of polytropic optimization}]
The minimum (resp. maximum) equilibrium solution concept of a continuous vector field that is defined over a closed and convex domain is the set of minimal (resp. maximal) elements of its linear vector polyorder.
\end{thesis}

We will refer to the problem of ``solving'' linear scalar and vector polyorders in the sense defined above as {\em polytropic optimization}. The minimal solutions of a polytropic optimization problem are precisely those points that lack directions of descent, and it is now easy to verify that strict local maxima, for example, do not qualify as solutions of a minimization problem. Our claim, however, is much stronger: {\em We claim that our solution concepts are a precise characterization of the intuitive notion of optimality.} To support our theses, we explore, in the next section, the precise relationship between our theory and the theories of nonlinear optimization and evolutionary games.

\if(0)

In this vein, note the following:
\begin{enumerate}

\item Since gradients fields are {\em irrotational}, systems that navigate scalar fields in their course of evolution by following a field's {\em flow lines} can neither cycle nor exhibit chaotic behavior.

\item That is generally not true, however, for general vector fields and it is precisely for this reason that the second fundamental thesis of polytropic optimization postulates that the minimal elements of a linear vector polyorder are the {\em equilibrium} solution concept of the corresponding vector field; devising a general solution concept for vector fields is an open question. 

\end{enumerate}

\fi

\section{Relationship with the theories of nonlinear optimization and evolutionary games}
\label{qowieurksdfjhgskdjgh}

Following preliminary results, we show, in this section, that minimal and maximal elements are necessarily critical elements, and, in this sense, our solution concept is an {\em equilibrium refinement concept}. Then we show that optimization theory's {\em strict local minima} as well evolutionary game theory's {\em evolutionarily stable states} are minimal. Therefore, our equilibrium refinement concept {\em encompasses} the most widely applied solution concepts in these theories. 

Thereafter, we explore the relationship between minimality and the concepts of {\em local minimum} and {\em neutrally stable state}. We provide examples showing that these concepts do not qualify as ``solution concepts'' in their respective theories and indeed they are not used as such. Our understanding from standard expositions of these theories is that local minima and neutrally stable states are equilibrium refinements that are {\em necessary} for optimality in that all optimal solutions of a nonlinear optimization problem should be local minima and that all optimal solutions of a population game should be neutrally stable. Indeed, evolutionary game theory's most general, to the extent of our knowledge, solution concept (one that encompasses the concept of an evolutionarily stable state) is that of an {\em evolutionarily stable set}~\cite{Thomas} (see also~\cite{Weibull}) all of whose members are neutrally stable. Nonlinear optimization theory is not equipped with an analogue of the evolutionarily stable set, but it is easy to devise one. We prove in this section that these quite general solution concepts are encompassed in the notion of minimality, which is an intuitively pleasing result. 

What's more, we show that minimality is a more general concept than all known prior to this paper solution concepts in a strong sense: Minimal elements are not necessarily local minima. The counterexample we provide is a mathematically interesting object to which we devote Section~\ref{aslkdjfhsdkjfhsdjf}.

\subsection{Preliminaries: Minimal and maximal elements}

In this section, we assume that $X$ is a closed and convex subset of $\mathbb{R}^m$, and that $c : X \rightarrow \mathbb{R}^m$ and $f : X \rightarrow \mathbb{R}$ are continuous. 

\begin{definition}
$x^*$ is $\preceq_c$-minimal if $\forall x \in X: \neg (x \prec_c x^*)$. It is $\preceq_c$-maximal if $\forall x \in X: \neg (x^* \prec_c x)$.
\end{definition}

The concepts of $\preceq_f$-minimality and $\preceq_f$-maximality have analogous definitions. The concepts of minimality and maximality have the following equivalent characterizations.

\begin{lemma}
\label{alskjdfjsdkfjdkjf}
Let $x^* \in X$ and, for any $x \in X$, let $x_\epsilon = \epsilon x^* + (1-\epsilon) x$.\\
(i) $x^*$ is $\preceq_c$-minimal if and only if 
\begin{align*}
\forall x: (\forall \epsilon \in [0,1]: x^* \cdot c(x_\epsilon) \leq x \cdot c(x_\epsilon)) \vee (\exists \epsilon \in [0,1]: x^* \cdot c(x_\epsilon) < x \cdot c(x_\epsilon)).
\end{align*}
(ii) $x^*$ is $\preceq_c$-maximal if and only if 
\begin{align*}
\forall x: (\forall \epsilon \in [0,1]: x^* \cdot c(x_\epsilon) \geq x \cdot c(x_\epsilon)) \vee (\exists \epsilon \in [0,1]: x^* \cdot c(x_\epsilon) > x \cdot c(x_\epsilon)).
\end{align*}
(iii) $x^*$ is $\preceq_f$-minimal if and only if
\begin{align*}
\forall x: (\forall \epsilon_1, \epsilon_2 \in [0,1]: \epsilon_1 < \epsilon_2 \Rightarrow f(x_{\epsilon_1}) \geq f(x_{\epsilon_2})) \vee (\exists \epsilon_1, \epsilon_2 \in [0,1] \mbox{ where } \epsilon_1 < \epsilon_2: f(x_{\epsilon_1}) > f(x_{\epsilon_2})).
\end{align*}
(iv) $x^*$ is $\preceq_f$-maximal if and only if
\begin{align*}
\forall x: (\forall \epsilon_1, \epsilon_2 \in [0,1]: \epsilon_1 < \epsilon_2 \Rightarrow f(x_{\epsilon_1}) \leq f(x_{\epsilon_2})) \vee (\exists \epsilon_1, \epsilon_2 \in [0,1] \mbox{ where } \epsilon_1 < \epsilon_2: f(x_{\epsilon_1}) < f(x_{\epsilon_2})).
\end{align*}
\end{lemma}

\begin{proof}
We only prove part $(i)$ as the proofs of the other parts are analogous.
\begin{align*}
x^* \mbox{ is minimal } &\Leftrightarrow \forall x: \neg(x \prec x^*)\\
&\Leftrightarrow \forall x: \neg((x \preceq x^*) \wedge \neg(x^* \preceq x))\\
&\Leftrightarrow \forall x: (x^* \preceq x) \vee \neg(x \preceq x^*)\\
&\Leftrightarrow \forall x: (\forall \epsilon \in [0,1]: x^* \cdot c(x_\epsilon) \leq x \cdot c(x_\epsilon)) \vee \neg(\forall \epsilon \in [0,1]: x \cdot c(x^*_\epsilon) \leq x^* \cdot c(x^*_\epsilon))\\
&\Leftrightarrow \forall x: (\forall \epsilon \in [0,1]: x^* \cdot c(x_\epsilon) \leq x \cdot c(x_\epsilon)) \vee (\exists \epsilon \in [0,1]: x \cdot c(x^*_\epsilon) > x^* \cdot c(x^*_\epsilon))\\
&\Leftrightarrow \forall x: (\forall \epsilon \in [0,1]: x^* \cdot c(x_\epsilon) \leq x \cdot c(x_\epsilon)) \vee (\exists \epsilon \in [0,1]: x^* \cdot c(x_\epsilon) < x \cdot c(x_\epsilon)).\qedhere
\end{align*}
\end{proof}

We are going to use the following proposition.

\begin{proposition}
\label{oeritueorituut}
$x^*$ is $\preceq_c$-minimal if and only if it is $\preceq_{-c}$-maximal.
\end{proposition}

\begin{proof}
By Lemma~\ref{alskjdfjsdkfjdkjf}, $x^*$ is $\preceq_{-c}$-maximal if and only if
\begin{align}
\forall x: (\forall \epsilon \in [0,1]: x^* \cdot (-c(x_\epsilon)) \geq x \cdot (-c(x_\epsilon))) \vee (\exists \epsilon \in [0,1]: x^* \cdot (-c(x_\epsilon)) > x \cdot (-c(x_\epsilon))).\label{balskdjfhdskjfh}
\end{align}
Furthermore, by the same lemma, $x^*$ is $\preceq_c$-minimal if and only if
\begin{align}
\forall x: (\forall \epsilon \in [0,1]: x^* \cdot c(x_\epsilon) \leq x \cdot c(x_\epsilon)) \vee (\exists \epsilon \in [0,1]: x^* \cdot c(x_\epsilon) < x \cdot c(x_\epsilon)).\label{cbaoiefhsdkjfh}
\end{align}
But~\eqref{balskdjfhdskjfh} and~\eqref{cbaoiefhsdkjfh} are equivalent.
\end{proof}

Finally, we state without proof the following proposition.

\begin{proposition}
\label{oeritueorituutf}
$x^*$ is $\preceq_f$-minimal if and only if it is $\preceq_{-f}$-maximal.
\end{proposition}

\subsection{Critical elements}

First we show that the set of minimal elements is a subset of the set of critical elements. 

\begin{definition}
$x^*$ is a critical element of $(X, c)$ if $\forall x \in X: x \cdot c(x^*) \geq x^* \cdot c(x^*)$.
\end{definition}

We are going to need the following lemma.

\begin{lemma}
\label{ksjdhfksdjfhjd}
If $y \cdot c(x) < x \cdot c(x)$, then there exists $\epsilon' \in [0,1]$ such that, $\forall \epsilon \in [\epsilon', 1]$, $y \cdot c(y_\epsilon) < x \cdot c(y_\epsilon)$ where $y_\epsilon = \epsilon x + (1-\epsilon) y$.
\end{lemma}

\begin{proof}
The proof is a simple implication of the continuity of $(y-x) \cdot c(y_\epsilon)$.
\end{proof}

\begin{theorem}
\label{lksajdfhskdjfhd}
If $x^*$ is $\preceq_c$-minimal, then $x^*$ is a critical element of $(X, c)$.
\end{theorem}

\begin{proof}
Let $x^*$ be minimal and suppose that there exists $x \in X$ such that $x \cdot c(x^*) < x^* \cdot c(x^*)$. Then, by Lemma \ref{ksjdhfksdjfhjd},  there exists an element in the convex hull of $x^*$ and $x$ that dominates $x^*$, which contradicts our assumption that $x^*$ is minimal.
\end{proof}

Note that the maximal elements of $\preceq_c$ are, in general, {\em not} a subset of the critical elements of $(X, c)$. However, they are a subset of the critical elements of $(X, -c)$, a fact that is a simple implication of Proposition~\ref{oeritueorituut}. Furthermore, note the following:
\begin{enumerate}

\item Since the maximal elements of $(X,c)$ are the minimal elements of $(X, -c)$, any theorem we prove below about minimal elements has an analogous statement for maximal elements.

\item The set of critical elements may have members that are neither minimal nor maximal. Consider, for example, $f(x) = x^2, x \in \mathbb{R}$. Viewing $f$ as a vector field, the origin is a critical element, but it is neither a minimal nor a maximal element (since the origin is dominated by every element of the left axis and dominates every element of the right axis).

\end{enumerate}

\subsection{Strict local minima and evolutionarily stable states}

\subsubsection{Linear vector polyorders}

\begin{definition}[See~\cite{PopulationGames}]
$x^*$ is an evolutionarily stable state of $(X, c)$ if there exists a neighborhood $O \subseteq X$ of $x^*$ such that $\forall x \in O - \{x^*\}: x^* \cdot c(x) < x \cdot c(x)$.
\end{definition}

\begin{theorem}
\label{qwpoeriuweoiru}
If $x^*$ is  an evolutionarily stable state of $(X, c)$, then $x^*$ is $\preceq_c$-minimal.
\end{theorem}

\begin{proof}
Let $x_\epsilon = \epsilon x^* + (1-\epsilon) x$, $\epsilon \in [0,1]$. By Lemma~\ref{alskjdfjsdkfjdkjf}, it suffices to show that for all $x \in X$ where $x \neq x^*$ there exists $\epsilon \in [0,1]$ such that $x^* \cdot c(x_\epsilon) < x \cdot c(x_\epsilon)$. By the assumption that $x^*$ is an evolutionarily stable state, for all $x \in O$ where $x \neq x^*$ there exists $\epsilon \in [0,1]$ such that $x^* \cdot c(x_\epsilon) < x_\epsilon \cdot c(x_\epsilon)$, which is also certainly true for all $x \in X$ since $O$ is a neighborhood; expanding $x_\epsilon$ and rearranging proves the theorem.
\end{proof}

\subsubsection{Linear scalar polyorders}

\begin{definition}
$x^*$ is a strict local minimum of $(X, f)$ if there exists a neighborhood $O \subseteq X$ of $x^*$ such that $\forall x \in O - \{x^*\}: f(x^*) < f(x)$.
\end{definition}

\begin{theorem}
\label{zxvsdflsdkljfhdslk}
If $x^*$ is a strict local minimum of $(X, f)$, then $x^*$ is $\preceq_f$-minimal.
\end{theorem}

\begin{proof}
Let $x_\epsilon = \epsilon x^* + (1-\epsilon) x$, $\epsilon \in [0,1]$. By Lemma~\ref{alskjdfjsdkfjdkjf}, it suffices to show that for all $x \in X$ where $x \neq x^*$ there exist $\epsilon_1, \epsilon_2 \in [0,1]$ where $\epsilon_1 < \epsilon_2$ such that $f(x_{\epsilon_1}) > f(x_{\epsilon_2})$. Suppose, for the sake of contradiction, that this is false. Then, there exists, $x \in O$ such that $\forall \epsilon_1, \epsilon_2 \in [0,1]$ with $\epsilon_1 < \epsilon_2$,  $f(x_{\epsilon_1}) \leq f(x_{\epsilon_2})$. But then $f(x) \leq f(x^*)$, which contradicts the assumption that $x^*$ is a strict local minimum.
\end{proof}

\subsection{Local minima, neutrally stable states, and evolutionarily stable sets}

\subsubsection{Linear vector polyorders}

Consider the definition of a {\em neutrally stable state} in evolutionary game theory.

\begin{definition}
$x^*$ is an neutrally stable state of $(X, c)$ if there exists a neighborhood $O \subseteq X$ of $x^*$ such that $\forall x \in O: x^* \cdot c(x) \leq x \cdot c(x)$.
\end{definition}

We may also analogously define the concept of a local minimum of a vector polyorder.

\begin{definition}
$x^*$ is a local minimum of $(X, \preceq_c)$ if there exists a neighborhood $O \subseteq X$ of $x^*$ such that $\forall x \in O: x^* \preceq_c x$.
\end{definition}

The concepts of local minimum of linear vector polyorders and of neutrally stable states are equivalent as the following theorem asserts.

\begin{theorem}
$x^*$ is a neutrally stable state of $(X, c)$ $\Leftrightarrow$ $x^*$ is a local minimum of $(X, \preceq_c)$.
\end{theorem}

\begin{proof}
Let $x^*$ be a neutrally stable state of $(X, c)$, and let $x \in O$. Let $x_\epsilon = \epsilon x^* + (1-\epsilon) x$. Then, for all $\epsilon \in [0,1]$, we have that  
\begin{align*}
x^* \cdot c(x_\epsilon) \leq x_\epsilon \cdot c(x_\epsilon) &\Leftrightarrow x^* \cdot c(x_\epsilon) \leq (\epsilon x^* + (1-\epsilon) x) \cdot c(x_\epsilon)\\
   &\Leftrightarrow (1-\epsilon) x^* \cdot c(x_\epsilon) \leq (1-\epsilon) x \cdot c(x_\epsilon)\\
   &\Leftrightarrow x^* \cdot c(x_\epsilon) \leq x \cdot c(x_\epsilon),
\end{align*}
and, therefore, $x^* \preceq_c x$. The reverse direction is analogous.
\end{proof}

Next we show that local minima are critical elements.

\begin{theorem}
\label{qwpoeiruweiru}
If $x^*$ is a local minimum of $(X, \preceq_c)$, then $x^*$ is a critical element of $(X, c)$.
\end{theorem}

\begin{proof}
Let $x^*$ be a local minimum of $(X, \preceq_c)$. Then there exists $O \subseteq X$ such that $\forall x \in O \mbox{ } \forall \epsilon \in [0,1]: x^* \cdot c(x_\epsilon) \leq x \cdot c(x_\epsilon)$. Therefore, for all $x \in O$, $x \cdot c(x^*) \geq x^* \cdot c(x^*)$. Suppose now there exists $x \in X-O$ such that $x \cdot c(x^*) < x^* \cdot c(x^*)$. Then, for all $\epsilon \in [0,1]$, $x_\epsilon \cdot c(x^*) < x^* \cdot c(x^*)$, which contradicts the previous implication that, for all $x \in O$, $x \cdot c(x^*) \geq x^* \cdot c(x^*)$.
\end{proof}

However, as, for example, shown in Figure~\ref{saldkjfnalxdkhjfgg}, a local minimum may not be minimal. 

\begin{figure}[tb]
\centering
\includegraphics[width=12cm]{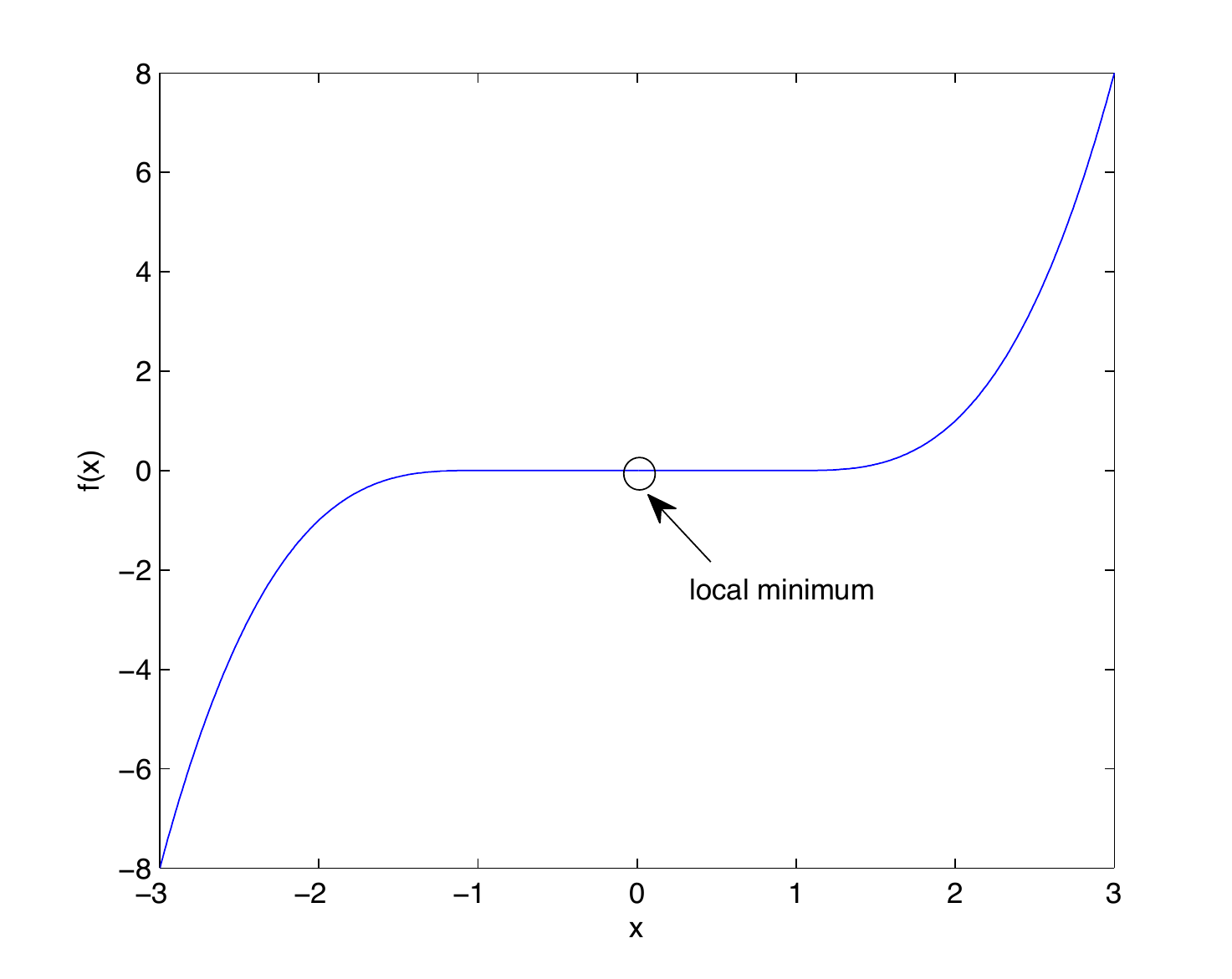}
\caption{\label{saldkjfnalxdkhjfgg}
An example showing a local minimum element that is not minimal.}
\end{figure}

Being a local minimum is a necessary condition in a generalization of the concept of the evolutionarily stable state, namely the {\em evolutionarily stable set}.

\begin{definition}[See~\cite{Weibull}]
$X^*$ is an evolutionarily stable set of $(X, c)$ if it is nonempty and closed and if for each $x^* \in X^*$ there exists a neighborhood $O(x^*) \subseteq X$ of $x^*$ such that $\forall x \in O(x^*): x^* \cdot c(x) \leq x \cdot c(x)$ with strict inequality if $x \not\in X^*$.
\end{definition}

Evolutionarily stable sets are minimal as the following theorem asserts.

\begin{theorem}
\label{laksjdhfsdkjfhdf}
Let $X^*$ be an evolutionarily stable set of $(X, c)$, and let $x^* \in X^*$. Then $x^*$ is $\preceq_c$-minimal.
\end{theorem}

\begin{figure}[tb]
\centering
\includegraphics[width=10cm]{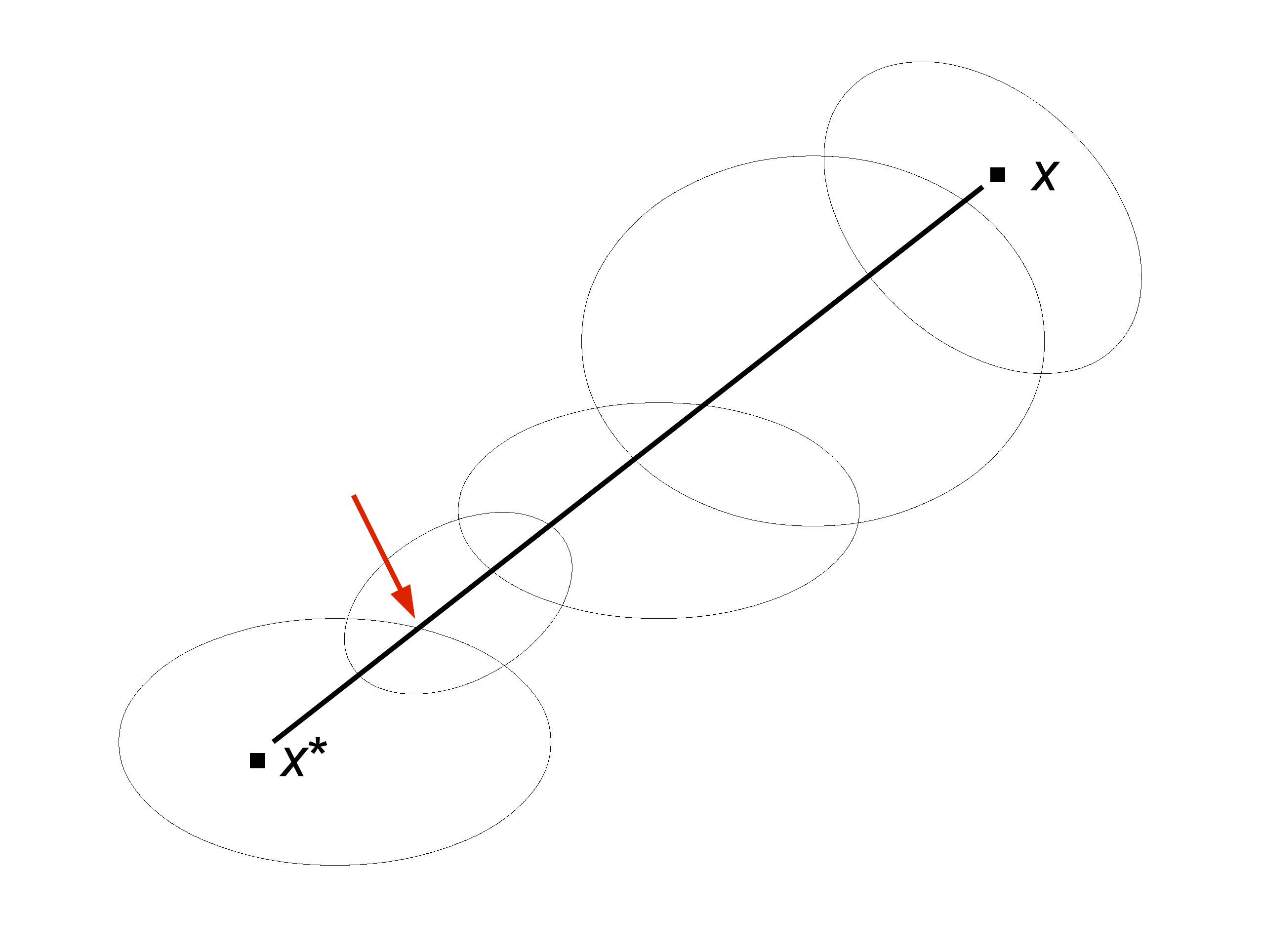}
\caption{\label{saldkjfnalxdkhjfgd}
Used in the proof of Theorem~\ref{laksjdhfsdkjfhdf}.}
\end{figure}

\begin{proof}
By Lemma~\ref{alskjdfjsdkfjdkjf}, it suffices to show that
\begin{align*}
\forall x \in X: (\forall \epsilon \in [0,1]: x^* \cdot c(x_\epsilon) \leq x \cdot c(x_\epsilon)) \vee (\exists \epsilon \in [0,1]: x^* \cdot c(x_\epsilon) < x \cdot c(x_\epsilon)).
\end{align*}
Note that it is trivially true (by the definition of an evolutionarily stable set) that 
\begin{align*}
\forall x \in O(x^*): (\forall \epsilon \in [0,1]: x^* \cdot c(x_\epsilon) \leq x \cdot c(x_\epsilon)) \vee (\exists \epsilon \in [0,1]: x^* \cdot c(x_\epsilon) < x \cdot c(x_\epsilon)).
\end{align*}
Therefore, take any $x \in X - O(x^*)$, and let $x_\epsilon = \epsilon x^* + (1-\epsilon) x$, $\epsilon \in [0,1]$. Let $x_{\epsilon'}$ be the point where $x_\epsilon$ intersects the boundary of $O(x^*)$ (pointed at by the arrow in Figure~\ref{saldkjfnalxdkhjfgd}), and note that to prove our claim we only need to consider the case that, for all $\epsilon \in [\epsilon',1]$, $x^* \cdot c(x_\epsilon) = x_{\epsilon'} \cdot c(x_\epsilon)$. Note further that in this case $x_{\epsilon'} \in X^*$ and consider $O(x_{\epsilon'})$. Let $x_{\epsilon''}$ be the point where $x_\epsilon$ intersects the boundary of $O(x_{\epsilon'})$ (in the direction toward $x$). Observe now that if there exists $\zeta \in [\epsilon', \epsilon'']$ such that $x_{\epsilon'} \cdot c(x_\zeta) < x_{\epsilon''} \cdot c(x_\zeta)$, then there exists $\epsilon \in [\epsilon', \epsilon'']$ such that $x^* \cdot c(x_\epsilon) < x \cdot c(x_\epsilon)$, and that if no such $\zeta$ exists we may similarly consider the boundary point $x_{\epsilon''}$ of $O(x_{\epsilon'})$ noting that then, for all $\zeta \in [\epsilon', \epsilon'']$, $x_{\epsilon'} \cdot c(x_\zeta) = x_{\epsilon''} \cdot c(x_\zeta)$, and, therefore, that $x_{\epsilon''} \in X^*$. Continuing in this way, we either obtain an $\epsilon \in [0,1]$ such that $x^* \cdot c(x_\epsilon) < x \cdot c(x_\epsilon)$ or we have that, for all $\epsilon \in [0,1]$, $x^* \cdot c(x_\epsilon) = x \cdot c(x_\epsilon)$. Since $x$ is arbitrary, the theorem is proven.
\end{proof}

\subsubsection{Linear scalar polyorders}

In the last part of this section, we explore the analogue of the concept of the evolutionarily stable set in scalar fields. To that end, consider the definition of a {\em local minimum} in nonlinear optimization.

\begin{definition}
$x^*$ is a local minimum of $(X, f)$ if there exists a neighborhood $O \subseteq X$ of $x^*$ such that $\forall x \in O: f(x^*) \leq f(x)$.
\end{definition}

We may also analogously define the concept of a local minimum of a scalar polyorder.

\begin{definition}
$x^*$ is a local minimum of $(X, \preceq_f)$ if there exists a neighborhood $O \subseteq X$ of $x^*$ such that $\forall x \in O: x^* \preceq_f x$.
\end{definition}

We have the following theorem.

\begin{theorem}
If $x^*$ is a local minimum of $(X, \preceq_f)$, then $x^*$ is a local minimum of $(X, f)$. 
\end{theorem}

\begin{proof}
Since $x^*$ is a local minimum of $(X, \preceq_f)$, for all $x \in O$, we have that, for all $\epsilon_1, \epsilon_2 \in [0,1]$ where $\epsilon_1 < \epsilon_2$, $f(x_{\epsilon_1}) \geq f(x_{\epsilon_2})$, which is certainly true if $\epsilon_1=0$ and $\epsilon_2 = 1$.
\end{proof}

However, the reverse direction is not generally true. For example, consider a scalar field $f:\mathbb{R}^2 \rightarrow \mathbb{R}$ that is obtained by rotating the function in Figure~\ref{saldkjfnalxdkhjfggp} around the $z$-axis (the figure shows the intersection of such a function with the plane $y=0$). It is easy to see in this example that the global minima of $f$ are not local minima of $\preceq_f$. (To see this note that the global minima of $f$ form a circle and take any two points on this circle.)

\begin{figure}[tb]
\centering
\includegraphics[width=12cm]{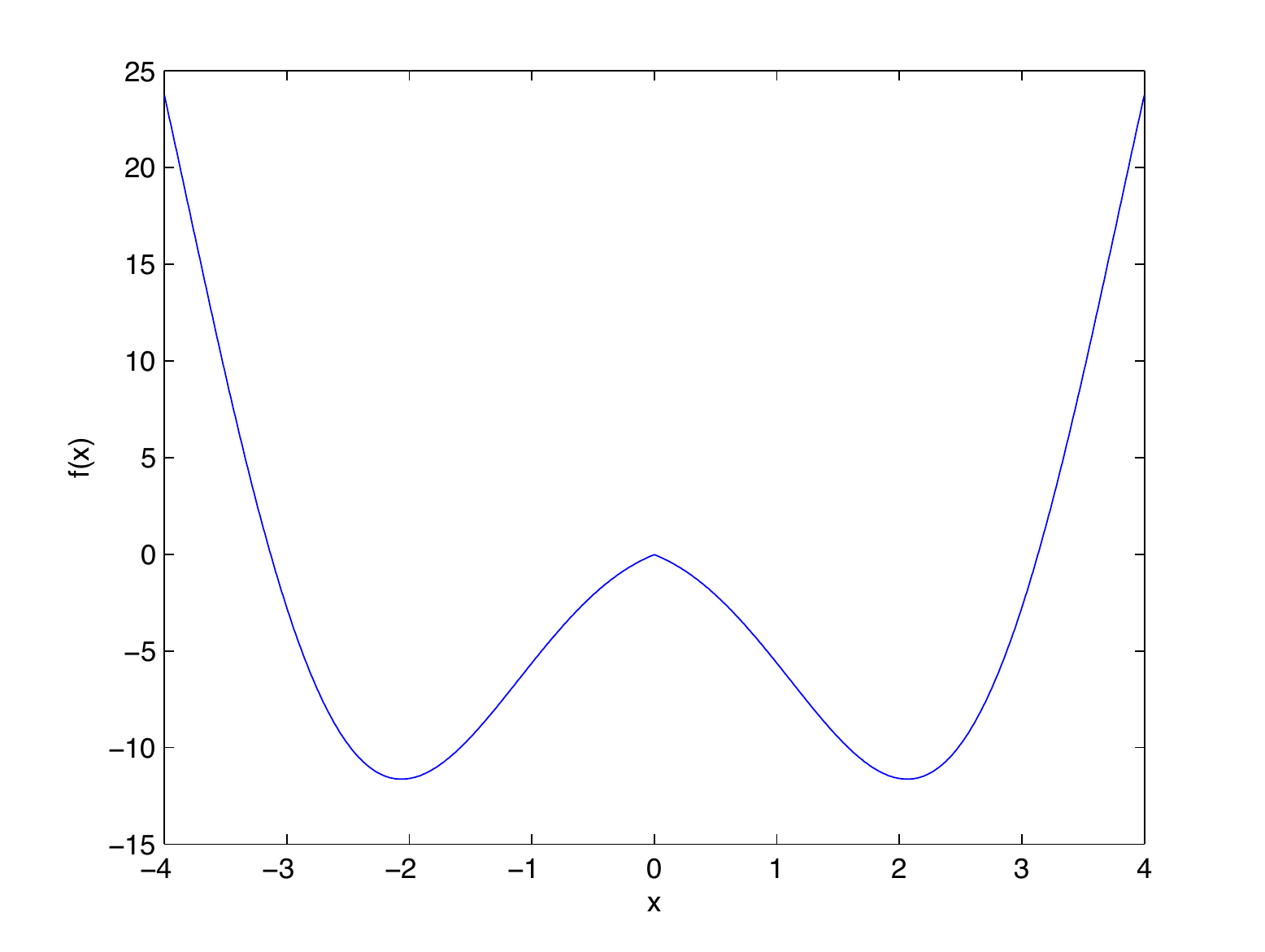}
\caption{\label{saldkjfnalxdkhjfggp}
Rotating this function around the vertical axis (at $x=0$) we obtain a scalar field whose global minima are not local minima of $\preceq_f$.}
\end{figure}

Since the local minima of a scalar field are not necessarily local minima of its corresponding scalar polyorder, we are motivated to strengthen the definition of an evolutionarily stable set in scalar fields by introducing a new concept, namely that of an {\em almost strictly minimal set}.

\begin{definition}
$X^*$ is an almost strictly minimal set of $(X, f)$ if it is nonempty and closed and if for each $x^* \in X^*$ there exists a neighborhood $O(x^*) \subseteq X$ of $x^*$ such that $\forall x \in O(x^*): f(x^*) \leq f(x)$ with strict inequality if $x \not\in X^*$.
\end{definition}

We have the following theorem, which is a stronger form of Theorem~\ref{laksjdhfsdkjfhdf}.

\begin{theorem}
Let $X^*$ be an almost strictly minimal set of $(X, c)$, and let $x^* \in X^*$. Then $x^*$ is $\preceq_f$-minimal.
\end{theorem}

\begin{proof}
The proof is analogous to that of Theorem~\ref{laksjdhfsdkjfhdf}, however, it is worth noticing the subtle differences in the beginning of the proof. By Lemma~\ref{alskjdfjsdkfjdkjf}, it suffices to show that
\begin{align*}
\forall x: (\forall \epsilon_1, \epsilon_2 \in [0,1]: \epsilon_1 < \epsilon_2 \Rightarrow f(x_{\epsilon_1}) \geq f(x_{\epsilon_2})) \vee (\exists \epsilon_1, \epsilon_2 \in [0,1] \mbox{ where } \epsilon_1 < \epsilon_2: f(x_{\epsilon_1}) > f(x_{\epsilon_2})).
\end{align*}
It is easy to show that the previous statement is true for all $x \in O(x^*)$. (The negation of the previous statement implies there exists a descent path starting at $x^*$, which contradicts the assumption that $x^*$ is a local minimum.) Therefore, let $x \in X - O(x^*)$, and let $x_\epsilon = \epsilon x^* + (1-\epsilon) x$, $\epsilon \in [0,1]$. Let $x_{\epsilon'}$ be the point where $x_\epsilon$ intersects the boundary of $O(x^*)$, and note that to prove our claim we only need to consider the case that, for all $\epsilon_1, \epsilon_2 \in [\epsilon',1]$ with $\epsilon_1 < \epsilon_2$ we have that $f(x_{\epsilon_1}) \geq f(x_{\epsilon_2})$,  in which case $f(x^*) = f(x_{\epsilon'})$ and, therefore,  $x_{\epsilon'} \in X^*$. From this point on the proof is completely analogous to that of Theorem~\ref{laksjdhfsdkjfhdf}.
\end{proof}

The converse of the previous theorem does not hold in general and, therefore, the concept of almost strict minimality does {\em not} characterize the concept of minimality. The next section is devoted to studying such a counterexample.

\section{Atypical solutions of polytropic optimization}
\label{aslkdjfhsdkjfhsdjf}

In this section, we study an optimization problem whose solution has atypical structure. The problem is that of optimizing $f(x) = x \sin(1/x), x \in \mathbb{R}$ (shown in Figure~\ref{saldkjfnalxdkhjfg}). Note that $f$ is differentiable everywhere except the origin where it is continuous. Note also that the solution of the scalar polyorder of $f$ is different from the solution of its vector polyorder, although both solutions have similar structure. 

\begin{figure}[tb]
\centering
\includegraphics[width=16cm]{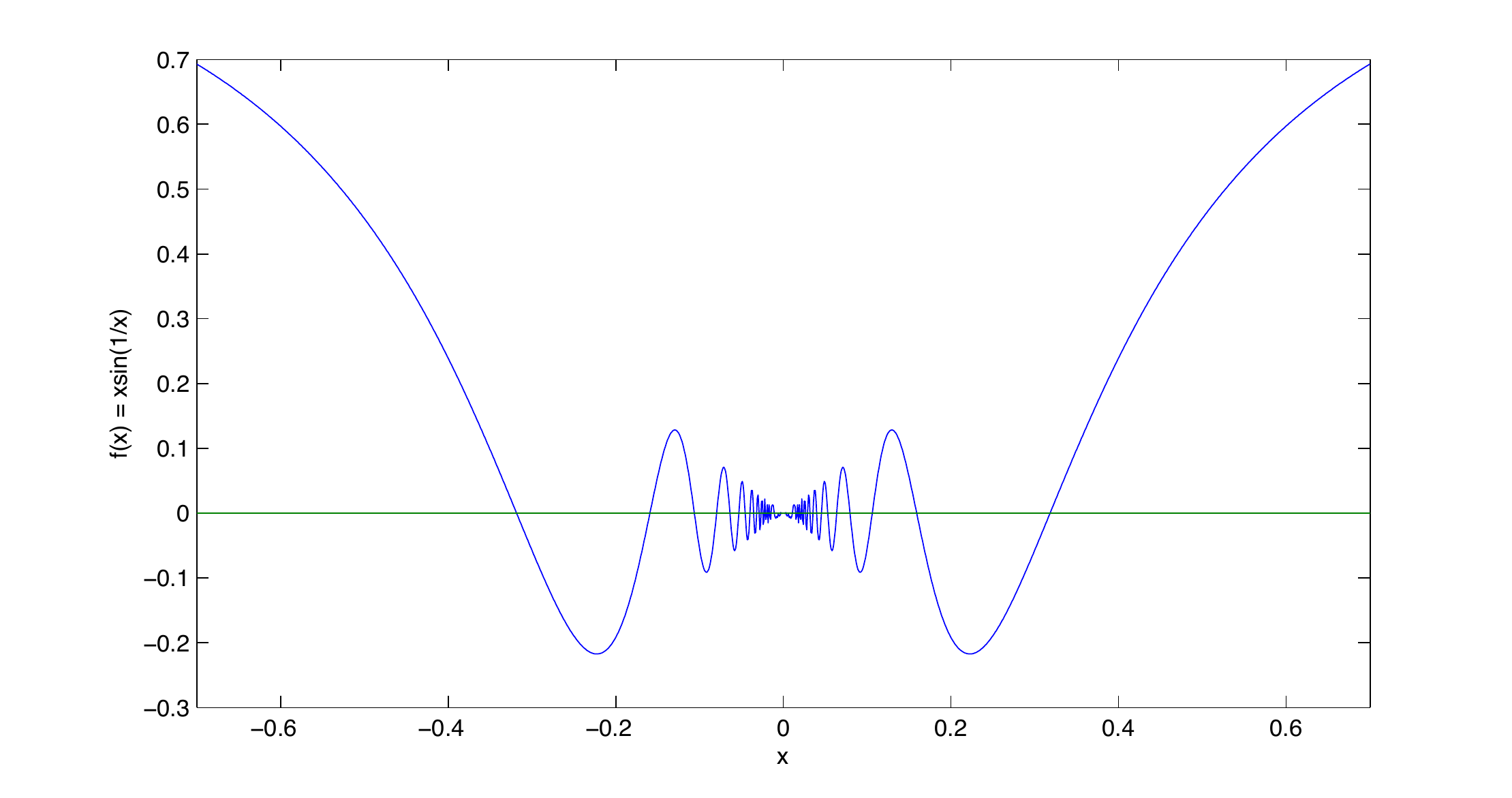}
\caption{\label{saldkjfnalxdkhjfg}
The graph of $f(x) = x \sin(1/x)$.}
\end{figure}

What is peculiar about the problem of optimizing $f$ is that the critical element $x=0$ is {\em both} minimal and maximal (in both polyorders) {\em without} being a local minimum or a local maximum (in either polyorder), however, it satisfies our definition of optimality because there exists neither a descent nor an ascent path on either side. 

In spite of the origin's atypical behavior, the optimal solutions of $f$ have a particularly attractive property once we view them as a {\em set}, namely, they are ``setwise locally dominant'' in that there exists a neighborhood of the set of minimal elements such that every element in this neighborhood is either minimal or it is dominated by a minimal element.

We also consider the dynamical system $dx/dt = - f(x)$. Looked at from a dynamical systems perspective, the origin is similarly peculiar: For one thing, it is neither an isolated point of the set of critical elements, as any neighborhood contains infinitely many critical elements, nor an interior point of that set. What's more, to approach the origin starting at {\em any} $x_0 \neq 0$ a physical system navigating $f$ has to ``reverse'' its evolution rule an infinite number of times, and on the grounds of this observation, we may, therefore, informally say that the origin would practically be ``spaced out'' for most physical systems. However, we show that the set of minimal elements is asymptotically stable as a {\em set}.

\subsection{Optimal solutions of $x\sin(1/x)$}

\begin{lemma}
\label{woeirutyeirtuy}
Let $f : \mathbb{R} \rightarrow \mathbb{R}$ be continuous, let $x^*$ be an isolated critical point of $f$, and suppose $f$ is differentiable in a neighborhood of $x^*$. Then, if $f'(x^*) > 0$, $x^*$ is minimal, whereas if $f'(x^*) < 0$, $x^*$ is maximal.
\end{lemma}

\begin{proof}
Suppose $f'(x^*) > 0$. It suffices to show, for all $x \in \mathbb{R} - \{x^*\}$, there exists $\epsilon \in [0,1]$ such that $x^* f(x_\epsilon) < x f(x_\epsilon)$. Expanding $f$ around $x^*$ we have that $f(x_\epsilon) = f'(x^*) (x_\epsilon-x^*) + o(|x_\epsilon-x^*|^2)$ and, therefore, that $(x^* - x) f(x_\epsilon) = (1/1-\epsilon) (x^* - x_\epsilon) f(x_\epsilon) \approx - (1/1-\epsilon) (x^*-x_\epsilon)^2 f'(x^*) < 0$ in a neighborhood of $x^*$. The proof that $x^*$ is maximal if $f'(x^*) < 0$ is analogous.
\end{proof}

\begin{theorem}
Viewing $f(x) = x\sin(1/x)$ as a vector field, the set of its minimal elements is $$X^*_{\min} = \{1/n \pi | n = 1,3,5, \ldots \} \cup \{ 0 \} \cup \{1/n \pi | n = -2,-4,-6, \ldots \},$$ and the set of its maximal elements is  $$X^*_{\max} = \{1/n \pi | n = 2,4,6, \ldots \} \cup \{ 0 \} \cup \{1/n \pi | n = -1,-3,-5, \ldots \}.$$
\end{theorem}

\begin{proof}
The critical points of $f$ are $\{ 0 \} \cup \{1/n \pi | n = \pm 1, \pm 2, \pm 3, \ldots \}$. Note now that every critical point except the origin is an isolated point and that $f$ is differentiable everywhere except at $0$. Therefore, we may use Lemma~\ref{woeirutyeirtuy} to characterize all critical points except the origin. To show that the elements of $X^*_{\min} - \{ 0 \}$ are minimal and that the elements of $X^*_{\max} - \{ 0 \}$ are maximal is then a matter of simple calculus. Consider now the origin. To show that the origin is both minimal and maximal it suffices to show that, for all $x \in \mathbb{R} - \{0\}$, there exists $\epsilon \in [0,1]$ such that $0 < x f(x_\epsilon) \Leftrightarrow \sin(1/x_\epsilon) > 0$ and $\epsilon' \in [0,1]$ such that $0 > x f(x_{\epsilon'}) \Leftrightarrow \sin( 1/x_{\epsilon'}) < 0$. Both properties follow by elementary properties of the sinusoidal function.
\end{proof}

\subsection{Minimal elements are setwise locally dominant}

\begin{proposition}
Every non-minimal element of the interval $[-1/\pi, + \infty)$ is dominated by a minimal element.
\end{proposition}

\begin{proof}
First we show that the ``rightmost'' minimal element, that is, minimal element $x^* = 1/\pi$, dominates every element in the interval $(1/\pi, +\infty)$. Letting $x \in (1/\pi, +\infty)$, it suffices to show that $(x^* - x) f(x) < 0$, which holds since $x^* < x$ and $f(x) > 0$. Next we show that the same minimal element dominates every element in the interval $(1/2\pi, 1/\pi)$. Letting $x \in (1/2\pi, 1/\pi)$, it suffices to show that $(x^* - x) f(x) < 0$, which holds since $x^* > x$ and $f(x) < 0$. Now we show that the minimal element $x^* = 1 / (2k+1) \pi, k =1,2,3,\ldots$ dominates every element in the intervals $(1 / (2k+1) \pi, 1/2k\pi)$ and $(1 / 2 (k+1) \pi, 1 / (2k+1) \pi)$. Again it is suffices to show that $(x^* - x) f(x) < 0$ where $x$ is any element in one of those intervals. If $x \in (1 / (2k+1) \pi, 1/2k\pi)$, then $x^* < x$ and $f(x) > 0$ (to see this note that since $x^*$ is minimal $f'(x^*) > 0$) whereas if $x \in (1 / 2 (k+1) \pi, 1 / (2k+1) \pi)$ then $x^* > x$ and $f(x) < 0$. Therefore, every element on the right halfline is either minimal or it is dominated by a minimal element. The proof for the left halfline is similar, noting that $x^* = -1/\pi$ is maximal, and, therefore, it is dominated by every element in the interval $(-\infty, -1/\pi)$.
\end{proof}

\subsection{Minimal elements are setwise asymptotically stable}

\subsubsection{Preliminaries}

Following~\cite{Hirsch-Smale}, consider the differential equation 
\begin{align}
\dot{x} = F(x)\label{alskdjfhsdkjfhd} 
\end{align}
where $F : X \subseteq \mathbb{R}^m \rightarrow \mathbb{R}^m$ and $X$ is open in $\mathbb{R}^m$. Let $x^*$ be an equilibrium of this equation (that is, an element of $X$ such that $F(x^*) = 0$). We call $x^*$ a {\em stable} equilibrium if for every neighborhood $O \subseteq X$ of $x^*$ there is a neighborhood $O' \subseteq X$ of $x^*$ such that every solution $x(t)$ with $x(0) \in O$ is defined and is in $O'$ for all $t > 0$. If $O'$ can be chosen so that in addition to the previous property, $\lim_{t \rightarrow +\infty} x(t) = x^*$, then $x^*$ is {\em asymptotically stable}.

\begin{theorem}[See~\cite{Hirsch-Smale}]
Let $x^*$ be an equilibrium of~\eqref{alskdjfhsdkjfhd}. Let $L: O \rightarrow \mathbb{R}$ be continuous function defined on a neighborhood $O \subseteq X$ of $x^*$, differentiable on $O - x^*$, such that $L(x^*) = 0$, $L(x) > 0$ if $x \neq x^*$, and $\dot{L} \leq 0$ in $O - x^*$. Then $x^*$ is stable. If also $\dot{L} < 0$ in $O - x^*$, then $x^*$ is asymptotically stable.
\end{theorem}

Following~\cite{PopulationGames}, let $X^* \subseteq X$ be a closed set, and call $O \subseteq X$ a neighborhood of $X^*$ if it open relative to $X$ and contains $X^*$. Call $X^*$ {\em Lyapunov stable} under~\eqref{alskdjfhsdkjfhd} if for every neighborhood $O$ of $X^*$ there exists a neighborhood $O'$ of $X^*$ such that every solution $\{x_t\}$ that start in $O'$ is contained in $O$: that is, $x_0 \in O'$ implies that $x_t \in O$ for all $t \geq 0$. $X^*$ is {\em attracting} if there is a neighborhood $\mathcal{O}$ of $X^*$ such that every solution that starts in $\mathcal{O}$ converges to $X^*$: that is, $x_0 \in \mathcal{O}$ implies that $\omega(\{x_t\}) \subseteq X^*$. $X^*$ is globally attracting if it is attracting with $\mathcal{O} = X$. Finally, the set $X^*$ is {\em asymptotically stable} if it is Lyapunov stable and attracting, and it is {\em globally asymptotically stable} if it is Lyapynov stable and globally attracting.

\begin{theorem}[See~\cite{PopulationGames}]
Let $X^* \subseteq X$ be closed, and let $\mathcal{O} \subseteq X$ be a neighborhood of $X^*$. Let $L: \mathcal{O} \rightarrow \mathbb{R}_+$ be $C^1$ with $L^{-1}(0) = X^*$. If $\dot{L}(x) \equiv \nabla L(x)' F(x) < 0$ for all $\mathcal{O} - X^*$, the $X^*$ is asymptotically stable under \eqref{alskdjfhsdkjfhd}. If in addition $\mathcal{O} = X$, then $X^*$ is globally asymptotically stable under \eqref{alskdjfhsdkjfhd}.
\end{theorem}

\subsubsection{Minimal elements are setwise asymptotically stable}

\begin{theorem}
\label{alskdjfhskdjfh}
$X^*_{\min}$, the set of minimal elements of $f(x) = x \sin(1/x), x \in \mathbb{R}$, is asymptotically stable under $\dot{x} = - x \sin(1/x)$.
\end{theorem}

In the proof of this theorem we are going to need the following lemma.

\begin{lemma}
\label{eprotiyurtoyiu}
Let $\dot{x} = - f(x)$. Then we have:\\
(i) $1/\pi$ is globally asymptotically stable in $(1/2\pi, +\infty)$.\\
(ii) $1/(2k+1)\pi, k = 1, 2, \ldots$ is globally asymptotically stable in $(1/2k\pi, 1/2(k+1)\pi)$.\\
(iii) $-1/2k\pi, k = 1, 2, \ldots$ is globally asymptotically stable in $(-1/(2k-1)\pi, -1/(2k+1)\pi)$.
\end{lemma}

\begin{proof} 
To prove part (i), let $L_{1/\pi}(x) = \int_{1/\pi}^x f(y) dy$. Then $L(1/\pi) = 0$, $L'(x) = f(x)$, and $\dot{L}(x) = - f^2(x)$. Therefore, $\dot{L}(x) < 0, x \neq 1/\pi$, which proves part (i). The proof of part (ii) is derived using a similar line of reasoning using $L_{1/(2k+1)\pi}(x) = \int_{1/(2k+1)\pi}^x f(y) dy$ and likewise for part (iii).
\end{proof}

\begin{proof}[Proof of Theorem~\ref{alskdjfhskdjfh}]
Since, by Lemma~\ref{eprotiyurtoyiu}, every solution that starts in $(-1/\pi, + \infty)$ remains in this set, $X^*_{\min}$ is Lyapunov stable. What's more, by the same lemma, every solution that starts in $(-1/\pi, + \infty)$ converges to $X^*_{\min}$. Therefore, $X^*_{\min}$ is asymptotically stable.
\end{proof}

\section{Related work}
\label{qwoeriwueroweiru}

To the extent of our knowledge, this is the first paper that proposes a general theory of equilibrium behavior in vector fields. Our theory builds upon various related ideas in the literature perhaps the most significant of which are the formalization of the variational inequality problem, the concept of invasion in evolutionary game theory and the concept of descent algorithms in nonlinear optimization theory, and last but not least order theory. It is difficult to trace the origins and evolution of these ideas in one paper (for example, optimization theory's necessary optimality conditions date as far back as 1637; see~\cite{Bertsekas}), and we defer from providing definitive references.

\subsection{The variational inequality problem}

The variational inequality problem captures the notion of {\em equilibrium} in its {\em generality} and represents, in our view, the first attempt to formalize a general theory of equilibrium behavior. However, our theory refines the concept of variational equilibrium in a strong sense. According to their variational formalization, equilibria are precisely those points that lack (a generalized notion of)  {\em directions of strict descent}. In contrast, minimal elements are precisely those points that lack (a generalized notion of)  {\em directions of descent}. Since points that lack directions of descent clearly lack directions of strict descent, minimal elements are a {\em variational equilibrium refinement concept}. Therefore, the concept of variational equilibrium is a necessary condition in any {\em fixed} (equilibrium) solution of a vector field, however, we have argued extensively that this concept cannot serve as a general solution concept in vector fields.

\subsection{Nonlinear optimization theory and game theory}

One of the main contributions of this paper is to explore the precise relationship between evolutionary game theory and nonlinear optimization theory: Drawing on the link that the variational inequality problem establishes between the equilibrium concepts of these theories, this is the first paper to establish a relationship between the evolutionarily stable state and the strict local minimum. However, our contribution extends beyond establishing this relationship to extending both theories through proposing an overarching solution concept (of which both the evolutionarily stable state and the strict local minimum are special cases).

At a more elementary conceptual level, we believe that evolutionary game theory's most significant contribution in the study of equilibrium behavior is the concept of {\em evolutionary stability}, and although there are many alternative formalizations of this concept, they have been shown to be {\em equivalent}. These formalizations are all based on the more elementary concept of {\em invasion}. Our theory is based instead on the elementary concept of {\em dominance}, which furnishes a solution concept that generalizes the concept of evolutionary stability in a strong sense.

Similarly, nonlinear optimization theory's formalizations are based on the elementary concept of a {\em direction of strict descent} whereas our theory is based instead on the concept of a {\em descent path} (that is germane to the concept of dominance). This change in perspective furnishes a solution concept that generalizes nonlinear optimization theory's solution concepts.

\subsection{Order theory}

The concept of a polyorder draws heavily on ideas from {\em order theory}, a discipline which is concerned, however, with the study of {\em transitive} binary relations (see, for example,~\cite{Ok}), whereas polyorders are generally not transitive. As discussed earlier, there is an increasing body of work on {\em acyclic} binary relations, a concept that relaxes the concept of transitivity. Although scalar polyorders are acyclic (Proposition~\ref{wpoeritueoiriut}), the rich structure of scalar polyorders warrants independent scrutiny.

\section*{Acknowledgments}

I would like to thank Jon Crowcroft for carefully listening many of my arguments.

\bibliographystyle{plain}
\bibliography{axiomatic}

\if(0)

\subsection{Preliminaries}

We start off with preliminary definitions in the theory of dynamical systems. First we introduce the concept of a {\em flow}, followed by the concept of a {\em gradient flow} ... Then we consider a particular instance of gradient flows, namely, first order flows.

\subsubsection{Flows}

The simplest algorithm to navigate a vector field $c : \mathbb{R}^m \rightarrow \mathbb{R}^m$ is to just ``go with the flow.'' To explore this idea we define the following (autonomous) dynamical system: 
\begin{align*}
\frac{dx}{dt} = \dot{x} = c(x).
\end{align*} 
Dynamical systems of this form are known as {\em flows} in the following sense: Any continuously differentiable path $t \mapsto p(t) \in \mathbb{R}^m$ satisfying the previous equation is called a {\em solution path}. Since for any solution path $p(t)$ we have that $\dot{p}(t) = c(p(t))$, the vector field $c$ is also called a {\em velocity field}. The collection of points in any solution path is called a {\em flow line}. Existence and uniqueness of solution paths for flows is a classical result in the theory of ordinary differential equations.

\subsubsection{Gradient flows}

Let $f :  \mathbb{R}^m \rightarrow  \mathbb{R}$ be a differentiable function whose gradient is Lipschitz continuous and set $c = - \nabla f$. Then $\dot{x} = c(x)$ is called a {\em gradient dynamical system}.

\subsubsection{First order flows}

A {\em first order flow} is an autonomous dynamical system $\dot{x} = c(x)$ where $c :  \mathbb{R} \rightarrow  \mathbb{R}$. It is easy to see that first order flows are also gradient flows.\\

{\em You need to talk somewhere about stability theory.}

\subsection{Atypical solutions of first order flows}

$...$

\fi

\if(0)

\section{Vector fields as a mathematical foundation}
\label{poiudsfgfkbwer}

Vector fields, that is, maps of the form $F: X \subseteq \mathbb{R}^m \rightarrow \mathbb{R}^m$, are the fundamental object of study in a variety of disciplines;  in some, the connection is apparent and well-established in the literature while in other disciplines it is perhaps less so. In this paper, we are concerned with four disciplines that vector fields have found applications in, namely, {\em dynamical systems theory}, {\em nonlinear optimization theory}, {\em game theory}, and {\em evolutionary game theory}.

\subsection{Dynamical systems theory}

The subject of dynamical systems theory is to study the dynamic {\em behavior} of systems, and, in particular, rules by which the state of systems evolves over time. Dynamical systems are of two types: {\em discrete} and {\em continuous}. Continuous dynamical systems evolve according to the solutions of a system of differential equations of the form $$\frac{dx}{dt} = \dot{x} = F(x)$$ where $F$ is a vector field called the {\em velocity field}. We may classify dynamical systems as being either {\em unconstrained} or {\em constrained}.\footnote{This terminology is not standard in dynamical systems theory, and it is borrowed from optimization theory.} In an unconstrained dynamical system, velocity vectors never point out of the space of possible system states (for example, if the state space is the entire Euclidean space) whereas in constrained systems velocity vectors are not such restricted. Constrained dynamical systems are typically studied using {\em projection dynamics} (for example, see~\cite{xxx, yyy}).
Since the connection between dynamical systems and vector fields is well-established in the literature we do not motivate this connection further and immediately turn our attention elsewhere.

\subsection{Nonlinear optimization theory}

Every vector field ``contains'' a scalar field since by the theorem of Helmholtz (also known as the {\em fundamental theorem of vector calculus}) every vector field admits a {\em Helmholtz-Hodge decomposition} as the sum of an {\em irrotational} and a {\em solenoidal} vector field, and irrotational vector fields admit a {\em potential function}. Scalar fields, that is, functions of the form $f: X \subseteq \mathbb{R}^m \rightarrow \mathbb{R}$, are the fundamental object of study of nonlinear optimization theory at least on the surface as in practice optimality conditions and optimization algorithms are typically derived and stated using the gradient of the optimization problem's objective function. Therefore, vector fields are also a fundamental object of study of nonlinear optimization theory, although the vector fields that are of concern here (called {\em gradient fields}) are generally more well-behaved than their more general counterparts.

\subsection{Game theory}

Myerson defines Game Theory as ``[T]he study of mathematical models of conflict and cooperation between intelligent rational decision-makers''~\cite{Myerson}. Game theory's most popular object of study is the {\em normal-form game}, which is a triple $$(I, (S_i)_{i \in I}, (u_i)_{i \in I}),$$ where $I$ is the set of players, $S_i$ is the set of pure strategies available to player $i$, and $u_i: S \rightarrow \mathbb{R}$ is the utility function of player $i$ where $S = \times_i S_i$ is the set of all strategy profiles (combinations of strategies). Normal-form games are {\em particular instances} of a more general concept, that of {\em population games} (for example, see~\cite{PopulationGames}), which are discussed below where the connection between normal-form games and vector fields becomes immediately apparent (an observation that has perhaps not received appropriate attention in the community).

\subsection{Evolutionary game theory}

The origins of evolutionary game theory are traced in biology~\cite{xxx, yyy} in contrast with game theory whose origins are traced in the study of economic behavior~\cite{xxx}. However, the mathematical formalisms of these disciplines are in fact quite closely related to each other, and our primary focus in this paper is on abstractions.

In this paper, we are concerned with {\em population games} of which normal form games are a special case (for example, see~\cite{PopulationGames}). A population game $\mathcal{G}$ is a pair  $(X, c)$. $X$, the game's {\em state space} or {\em strategy profile space}, has product form, i.e., $X = X_1 \times \cdots \times X_n$, where the $X_i$'s are simplexes and $i=1,\ldots,n$ refers to a {\em player position} (or {\em population}). To each player position corresponds a set $S_i$ of $m_i$ {\em pure strategies} and a {\em mass} $\omega_i > 0$ (the population mass). The {\em strategy space} $X_i$ of player position $i$ has the form
\begin{align*}
X_i = \left\{ x \in \mathbb{R}^{m_i}  \bigg| \sum_{j \in S_i} x_j  = \omega_i \right\}.
\end{align*}
We refer to the elements of $X$ as {\em states} or {\em strategy profiles}. Each strategy profile $x \in X$ can be decomposed into a vector strategies, i.e., $x = (x_1,\ldots,x_n)'$. Let $m = \sum_i m_i$. $c: X \rightarrow \mathbb{R}^m$, the game's {\em cost function}, maps $X$, the game's state space, to vectors of {\em costs} where each position in the vector corresponds to a pure strategy. We assume that $c$ is {\em continuous}.

Population games are the first fundamental object of study of {\em evolutionary game theory}, noting that the origins of this theory are traced in biology~\cite{xxx} in contrast with game theory whose origins are traced in the study of economic behavior~\cite{xxx}. However, the mathematical formalisms of these disciplines are in fact quite closely related to each other as noted above, and our primary focus in this paper is on abstractions. In this vein, we will abstract population games (and, therefore, also normal-form games) as {\em vector fields} (that is, maps of the form $F: X \subseteq \mathbb{R}^m \rightarrow \mathbb{R}^m$). 

The second fundamental object of study of evolutionary game theory is the {\em evolutionary dynamics} of {\em revision protocols}~\cite{PopulationGames}. Revision protocols are simple procedures that population members employ in their effort to ``navigate'' a population game. (This is again in contrast with game theory whose assumption is that the ability to navigate a game is derived from introspection and rationality). Revision protocols are studied as evolution rules of {\em dynamical systems}.

\section{Critical points are not a solution concept}
\label{laskjdfhskjdhf}

In the finite-dimensional variational inequality problem (for example, see~\cite{xxx}), we are given a set $X \subseteq \mathbb{R}^m$, which is typically assumed to be nonempty, closed, and convex, and a vector field $c: \mathbb{R}^m \rightarrow \mathbb{R}^m$, and our objective is to find a vector $x^* \in X$ such that
\begin{align*}
(x - x^*) \cdot c(x^*) \geq 0, \forall x \in X.
\end{align*}
We will refer to the solutions of the variational inequality problem as the {\em critical points} or {\em critical elements} of the vector field $X$. To justify this terminology we observe that if $X$ is the entire $m$-dimensional space, then its critical points are precisely those points where the vector field vanishes. Therefore, the equilibria of an unconstrained dynamical system are such points as are the critical points of an unconstrained optimization problem. What's more the equilibria of a projected dynamical system~\cite{xxx}, the critical points of a constrained nonlinear optimization problem (see, for example,~\cite{xxx}), and the Nash equilibria of a population game (see, for example,~\cite{PopulationGames}) are all critical points of the variational inequality problem. However, critical points do not qualify as a solution concept in any of the aforementioned disciplines.

\subsection{Dynamical systems}

Dynamical systems theory has found many applications in various disciplines such as control theory, fluid mechanics, and population dynamics. Scientists working on these and other disciplines and being interested in understanding the systems they are studying are not content with an equilibrium analysis but almost always go beyond that to do stability analysis.

\subsection{Nonlinear optimization}

$f(x) = x^3$

\subsection{Game theory}

Talk about Daskalakis's result.

\subsection{Evolutionary game theory}

$...$

\section{Strict local minima}
\label{rtweiruteirut}

The concept of a critical point has been refined in all of the aforementioned disciplines.

\subsection{Dynamical systems}

Perhaps the most prominent refinement of the concept of a critical point in dynamical systems theory is that of {\em Lyapunov stability}.

\subsection{Nonlinear optimization}

Talk about convex optimization. Strict local minima are meaningful if the problem is convex and that is the most well-developed theory in nonlinear optimization. Argue that there are many meaningful problems that are not convex...

\subsection{Evolutionary game theory}

ESS... (Talk about the alternative characterizations of the ESS from Sandholm.)

\section{Strict local minima are not a solution concept}
\label{woeirtoeriut}

They are a solution concept in a convex problem. When the problem is not convex it is not just the problematic case that you may get multiple strict local minima. Show the picture that I had sent to Jon. Argue that the vector field generated by this picture does not give an asymptotically stable equilibrium Start talking about ESSets.

\section{Evolutionarily stable sets}

$...$

In the following section, we provide an alternative characterization of evolutionarily stable sets that corroborates the case, these sets are definitive solution concept for a general vector field.

\section{Polytropic optimization}
\label{zxjhfaoifdfjsndk}

In this section, we propose an {\em equilibrium solution concept} for general vector fields. At the heart of this solution concept lies a {\em preference relation} on the elements of the vector field. First we present background material on preference relations, followed by the solution concept we propose.

\subsection{Preference relations}

Let $X$ be a nonempty set. A subset $R$ of $X \times X$ is called a {\em binary relation} on X. If $(x,y) \in R$, we write $xRy$, and if $(x,y) \not\in R$, we write $\neg xRy$. $R$ is called {\em reflexive} if $xRx$ for each $x \in X$ and {\em complete} if for each $x,y \in X$ either $xRy$ or $yRx$. It is called {\em symmetric} if, for any $x, y \in X$, we have that $xRy \Rightarrow yRx$, and {\em antisymmetric} if, for any $x, y \in X$, $xRy \wedge yRx \Rightarrow x=y$. Finally, $R$ is called {\em transitive} if for any $x, y, z \in X$ we have that $xRy$ and $yRz$ imply $xRz$.

Let $xPy \Leftrightarrow xRy \wedge \neg yRx$ and $xIy \Leftrightarrow xRy \wedge yRx$. Then $P$ and $I$ are also binary relations on X where $P \subset R$ and $I \subset R$. $P$ is called the asymmetric part of $R$ and $I$ is called its symmetric part. If $P$ is antisymmetric, then $R$ is called a {\em preference relation}.

Preference relations are a fundamental object of study of {\em microeconomic theory}. The most elementary problem in this theory is that of {\em individual decision making} (see, for example,~\cite{xxx}) whose starting point is a set of alternatives (say $X$) among which an individual must choose. A fundamental tenet of microeconomic theory (since the definitive work of Von Neumann and Morgernstern~\cite{xxx}) is that individual preferences are {\em rational}.

\begin{definition}
The preference relation $R$ is {\em rational} if it is {\em complete} and {\em transitive}. 
\end{definition}

The study of transitive preference relations is the subject of {\em order theory} (for example, see~\cite{xxx}); although this theory is of considerable intellectual merit, in the rest of this paper, we will be concerned with preference relations that are generally non-transitive.

If a preference relation is not transitive, one would hope that it is perhaps acyclic.

\begin{definition}
The preference relation $R$ is {\em acyclic} if there is no finite set $\{ x_1, \ldots, x_n \} \subset X$ such that $x_i R x_{i+1}$ for $i = 1,\ldots,n-1$ and such that $x_n R x_1$. 
\end{definition}

There is a growing theory of acyclic preference relations (for example, see~\cite{xxx, yyy}), however, the preference relations we will be concerned with are in general not even acyclic.

\begin{definition}
Let $R$ be a preference relation on $X$ and let $x, y \in X$. If $xPy$, where $P$ is the asymmetric part of $R$, we say that $x$ {\em dominates} $y$. We further say that an element of $X$ is {\em minimal} (or {\em undominated}) if there is no element in $X$ that dominates it. We, finally, say that an element of $X$ is {\em maximal} if there no element in $X$ that it dominates.
\end{definition}

\subsection{An equilibrium solution concept for vector fields}

Let $c: X \rightarrow \mathbb{R}^m$ be a vector field where $X$ is a convex subset of $\mathbb{R}^m$. Let $x, y \in X$, and let $y_\epsilon = \epsilon x + (1-\epsilon) y$. Let $\preceq$ be a preference relation on $X$ such that 
\begin{align*}
x \preceq y \Leftrightarrow \forall \epsilon \in [0,1]: x \cdot c(y_\epsilon) \leq y \cdot c(y_\epsilon).
\end{align*} 
This preference relation can be understood as follows: If $x \preceq y$, the projection of $c(y_\epsilon)$ on the line segment connecting $x$ and $y$ does not point away from $x$ (that is, it either points toward $x$ or it vanishes). This preference relation is neither transitive nor acyclic. However, $\prec$, the asymmetric part of $\preceq$, is antisymmetric, and, therefore, $\preceq$ is indeed a preference relation. 

Let $x^* \in X$, let $O \subseteq X$ be a neighborhood of $x^*$  and consider the following definitions:
\begin{definition}
$x^*$ is a critical element of $\preceq$ if $\forall x \in X: x \cdot c(x^*) \geq x^* \cdot c(x^*)$.
\end{definition}

Note that according to this definition the problem of finding the critical elements of a vector field is a variational inequality problem, whose solutions we have previously argued are not the sought for solution concept. We may now understand why this is so: {\em The solutions of a variational inequality problem are impervious to the behavior of the vector field in their neighborhood.}

\begin{definition}
$x^*$ is a local minimum of $\preceq$ if if $\exists O$ such that  $\forall x \in O: x^* \preceq x$.
\end{definition}


Similarly, we may now understand why local minima are not the sought for solution concept either: {\em In the neighborhood of a local minimum the nearby vectors do not point away from this minimum, but they are not guaranteed to point toward the minimum.}

\begin{definition}
$x^*$ is a minimal element of $\preceq$ if $\forall x \in X: \neg (x \prec x^*)$.
\end{definition}

In contrast to critical elements that are impervious to the behavior of the vector field in their neighborhood, minimal elements critically depend on this behavior.

\begin{definition}
$x^*$ is a strict local minimum of $\preceq$ if $\exists O$ such that $\forall x \in O: x^* \prec x$.
\end{definition}


Furthermore, let
\begin{align*}
\mbox{CE}(X, c) &= \{ x^* | x^* \mbox{ is a critical element of } \preceq \}\\
\mbox{LM}(X, c) &= \{ x^* | x^* \mbox{ is a local minimum of } \preceq \}\\
\mbox{MIN}(X, c) &= \{ x^* | x^* \mbox{ is a minimal element of } \preceq \}\\
\mbox{SLM}(X, c) &= \{ x^* | x^* \mbox{ is a strict local minimum of } \preceq \}
\end{align*}

We have the following theorem.
\begin{theorem}
$\mbox{SLM}(X, c) \subseteq \mbox{MIN}(X, c) \subseteq \mbox{LM}(X, c) \subseteq \mbox{CE}(X, c).$
\end{theorem}

\begin{proof}
$...$
\end{proof}

\begin{definition}
$x^*$ is an almost strict local minimum of $\preceq$ if $\exists O$ such that $\forall x \in O/\mbox{MIN}(X, c): x^* \prec x$ and $\forall x \in O \cap \mbox{MIN}(X, c): x^* \sim x$.
\end{definition}

Let
\begin{align*}
\mbox{ASLM}(X, c) &= \{ x^* | x^* \mbox{ is an almost strict local minimum of } \preceq \}.
\end{align*}

We call the problem of finding the minimal elements of $\preceq$ {\em polytropic optimization} whose fundamental thesis is that these elements are the {\em solution concept} of a general vector field.

\section{Summary and further evidence}

Provide a summary of the argument you made up to here.

\subsection{Nonlinear optimization}

$...$ Here we have *the* solution concept...\\

Rewrite the results as you wrote them above (in polytropic optimization) for nonlinear optimization...\\

Make the case that the problem is solved -- we finally have a solution concept for nonlinear optimization...\\

How do you make the case? Well, there are no periodic orbits or something more complicated... Well, there should be a theorem about the $\omega$-limit sets being simple---use that theorem to say that the $\omega$-limit sets are not more complicated than equilibria---since we are proposing an equilibrium solution concept we are fine here.

\subsection{Evolutionary game theory}

$...$

\subsection{Dynamical systems}

$...$

\section{Discussion}

Talk about game theory's solution concept.

\bibliographystyle{plain}
\bibliography{axiomatic}

\newpage

\appendix

\section{Appendix to Section~\ref{zxjhfaoifdfjsndk}}
\label{qwoeiruwoeiirues}

\begin{definition}
\label{xcnbvlzfjv}
We say that $x$ weakly invades $y$ if $x \cdot c(y) \leq y \cdot c(y)$. If the inequality is strict, we say that $x$ invades $y$. If $\forall \epsilon \in [0,1], x \cdot c(x_\epsilon) \leq y \cdot c(x_\epsilon)$, we say that $x$ weakly dominates $y$. If the inequality is strict, we say that $x$ dominates $y$. 
\end{definition}

Let $\mathcal{G}=(X,c)$ be a population game, $x,y \in X$, $x_\epsilon = (1-\epsilon) x + \epsilon y$, $\epsilon \in [0,1]$, and $\delta(\epsilon) = (x-y) \cdot c(x_\epsilon)$.

\begin{lemma}
\label{qwerxmcbpdkfjhbdfg}
If $x$ dominates $y$, then $x \prec_\mathcal{G} y$.
\end{lemma}

\begin{lemma}
\label{slakfdjdskfj}
If $x$ invades $y$, then there exists $0 \leq \epsilon < 1$ such that $x_\epsilon$ that dominates $y$.
\end{lemma}

\begin{proof}
By assumption, $\delta(1) < 0$, and, therefore, there exists $0 \leq \epsilon < 1$ such that, for all $ \epsilon \leq \epsilon' \leq 1$, $\delta(\epsilon') < 0$, which completes the proof.
\end{proof}

\begin{lemma}
\label{weporitkjfhgdfpjgn}
If, for all $0 \leq \epsilon < 1$, $x_\epsilon$ does not dominate $y$, then $x$ does not invade $y$.
\end{lemma}

\begin{proof}
If $x$ invades $y$, then by Lemma~\ref{slakfdjdskfj}, there exists such $x_\epsilon$ that dominates $y$, a contradiction.
\end{proof}

\begin{lemma}
\label{oweiurteroitueroitu}
If $x$ strongly dominates $y$, then,  for all $\epsilon \in [0,1)$, $x_\epsilon$ strongly dominates $y$.
\end{lemma}

\begin{proof}

The proof is easy by picture.

\begin{align*}
x_\epsilon \cdot c(x_\epsilon) \leq y \cdot c(x_\epsilon) &\Leftrightarrow ((1-\epsilon) x + \epsilon y) \cdot c(x_\epsilon) \leq y \cdot c(x_\epsilon)\\
   &\Leftrightarrow (1-\epsilon) x \cdot c(x_\epsilon) \leq (1-\epsilon) y \cdot c(x_\epsilon)\\
   &\Leftrightarrow x \cdot c(x_\epsilon) \leq y \cdot c(x_\epsilon)
\end{align*}
\end{proof}

\begin{lemma}
$\mbox{{\em LM}}(X, c) \subseteq \mbox{{\em CE}}(X, c)$. 
\end{lemma}

\begin{proof}
Let $x^*$ be a local minimum of $\preceq$. Then there exists $O \subseteq X$ such that $\forall x \in O \mbox{ } \forall \epsilon \in [0,1]: x^* \cdot c(x_\epsilon) \leq x \cdot c(x_\epsilon)$. Suppose there exists $y \in X$ such that $y \cdot c(x^*) < x^* \cdot c(x^*)$. Then, by Lemma~\ref{slakfdjdskfj}, there exists $\epsilon \in [0,1)$ such that $(1-\epsilon)y+\epsilon x^*$ strongly dominates $x^*$. Therefore, by Lemma~\ref{oweiurteroitueroitu}, there exists an element of $X$ in $O$ that strongly dominates $x^*$, which contradicts our assumption.
\end{proof}

\begin{lemma}
We have the following equivalent characterizations of minimality and maximality:
\begin{align*}
x^* \in \mbox{{\em MIN}}(X, c) &\Leftrightarrow \forall x: (\forall \epsilon \in [0,1]: x^* \cdot c(x_\epsilon) \leq x \cdot c(x_\epsilon)) \vee (\exists \epsilon \in [0,1]: x^* \cdot c(x_\epsilon) < x \cdot c(x_\epsilon))\\
x^* \in \mbox{{\em MAX}}(X, c) &\Leftrightarrow \forall x: (\forall \epsilon \in [0,1]: x^* \cdot c(x_\epsilon) \geq x \cdot c(x_\epsilon)) \vee (\exists \epsilon \in [0,1]: x^* \cdot c(x_\epsilon) > x \cdot c(x_\epsilon)).
\end{align*}
\end{lemma}

\begin{proof}
\begin{align*}
x^* \in \mbox{MIN}(X, c) &\Leftrightarrow \forall x: \neg(x \prec x^*)\\
&\Leftrightarrow \forall x: \neg((x \preceq x^*) \wedge \neg(x^* \preceq x))\\
&\Leftrightarrow \forall x: (x^* \preceq x) \vee \neg(x \preceq x^*)\\
&\Leftrightarrow \forall x: (\forall \epsilon \in [0,1]: x^* \cdot c(x_\epsilon) \leq x \cdot c(x_\epsilon)) \vee \neg(\forall \epsilon \in [0,1]: x \cdot c(x^*_\epsilon) \leq x^* \cdot c(x^*_\epsilon))\\
&\Leftrightarrow \forall x: (\forall \epsilon \in [0,1]: x^* \cdot c(x_\epsilon) \leq x \cdot c(x_\epsilon)) \vee (\exists \epsilon \in [0,1]: x \cdot c(x^*_\epsilon) > x^* \cdot c(x^*_\epsilon))\\
&\Leftrightarrow \forall x: (\forall \epsilon \in [0,1]: x^* \cdot c(x_\epsilon) \leq x \cdot c(x_\epsilon)) \vee (\exists \epsilon \in [0,1]: x^* \cdot c(x_\epsilon) < x \cdot c(x_\epsilon)).\qedhere
\end{align*}
\end{proof}

\begin{lemma}
$\mbox{{\em MIN}}(X, c)/\mbox{{\em MAX}}(X, c) \subseteq \mbox{{\em LM}}(X, c)$. 
\end{lemma}

\begin{proof}
Let $x^*$ be a minimal element of $\preceq$ that is not also maximal. It suffices to show that, for all $x \in X$, there exists $\epsilon' > 0$ such that, for all $\epsilon \in [0, \epsilon']$, $x^* \cdot c(x_\epsilon) \leq x \cdot c(x_\epsilon)$ where $x_\epsilon = (1-\epsilon) x^* + \epsilon x$. By the minimality of $x^*$, we have
\begin{align*}
\forall x, \neg (x \prec x^*) &\Leftrightarrow x^* \preceq x \vee \neg(x \preceq x^*)
\end{align*}
\begin{align*}
\neg(x \preceq x^*) &\Leftrightarrow \neg(\forall \epsilon \in [0,1]: x \cdot c(x^*_\epsilon) \leq x^* c(x^*_\epsilon))\\
&\Leftrightarrow \exists \epsilon \in [0,1]: x \cdot c(x^*_\epsilon) > x^* c(x^*_\epsilon)
\end{align*}
\end{proof}

\begin{lemma}
$\mbox{{\em SLM}}(X, c) \subseteq \mbox{{\em MIN}}(X, c)$. 
\end{lemma}

\begin{proof}
Let $x^*$ be a strict local minimum of $\preceq$, and suppose $x^*$ dominates a neighborhood $O$. Suppose there exists $x \in X$ such that $x \prec x^*$. Then $x \preceq x^*$, which is the same as saying that $\forall \epsilon \in [0,1]: x \cdot c(x^*_\epsilon) \leq x^* c(x^*_\epsilon)$. Let $\epsilon' \in (\epsilon,1]$ be such that $x^*_{\epsilon'} \in O$. Then $x^* \prec x^*_{\epsilon'}$, which implies that there exists $\epsilon'' \in (\epsilon, 1]$ such that $x \cdot c(x^*_{\epsilon''}) > x^* \cdot c(x^*_{\epsilon''})$, which contradicts our assumption that $x \prec x^*$, and the lemma is proven.
\end{proof}

\fi

\if(0)

\section{Introduction (old)}

Consider a vector field $c: X \rightarrow \mathbb{R}^m$ where $X \subseteq \mathbb{R}^m$ is convex. Various problems in various scientific disciplines can be modeled as problems of finding the {\em nodes} of this vector field where a node is either a {\em source} or a {\em sink}. (Talk about how these terms are used in the literature...)

Although these terms are widely used, they lack precise definitions, which, in this paper, we give.

We call the problem of finding the nodes of a vector field {\em polytropic optimization}, and give necessary and sufficient conditions for a vector to be a node (that is, {\em optimal}). 

 We argue that our definitions correspond with how these terms are used in the literature.

\subsection{Vector fields}

$...$

\subsection{Polytropic optimization}

Polytropic optimization is the problem of finding the sources and sinks of a vector field.

-----

$$\forall y: \neg (y \prec x) \wedge \neg (x \prec y)$$

$$(y \cdot c(x) > x \cdot c(x) \vee x \cdot c(y) \leq y \cdot c(y)) \wedge (x \cdot c(y) > y \cdot c(y) \vee y \cdot c(x) \leq x \cdot c(x))$$

$$(y \cdot c(x) > x \cdot c(x) \wedge (x \cdot c(y) > y \cdot c(y)) \vee (x \cdot c(y) = y \cdot c(y))  \wedge y \cdot c(x) = x \cdot c(x))$$

This means that if I want an element $x$ to be both minimal and maximal at the same time, then for all $y$ the above must be true.

\subsection{Nonlinear optimization}

$...$

\subsection{Noncooperative game theory}

$...$

The fundamental problem in Noncooperative Game Theory is how a society of individuals can optimally decide among a set of alternatives: Consider a society of $n$ individuals, suppose that each individual has a {\em state}, let $X$ be the set of possible states for the individual, and suppose further that each individual $i$ has the ability to influence the state she is in by selecting a {\em strategy} from a set $X_i$. Note, however, that the state $i$ is in is influenced not only by the strategy she chooses but also by the strategies that the other members of society choose. Our basic assumption is that {\em all individuals must simultaneously be in the same state,} which can't be influenced exogenously. This implies that set $X$ of possible states is precisely the Cartesian product $X_1 \times \cdots \times X_n$.

Although we require that individuals be in the same state, not every individual values this state equally, and we may say that the {\em preferences} of the individuals over the possible states vary. We are going to assume that the individuals are {\em rational} in that the preferences of each individual over the possible states can be represented by a {\em utility function}.\footnote{We use the term `rational' in the sense that preferences are complete and transitive. In that sense, if the utility function is understood as a {\em fitness function}, as is customary in the applications of this model to Biology, then the corresponding population is implicitly understood to be rational.} If we further assume that individuals are {\em selfish} in that they pursue maximizing their respective utility functions, then society faces a decision problem that is the object of study of Noncooperative Game Theory.

Consider further the following thought experiment: Suppose that on every occasion in which the members of society choose a state, it is {\em as if} society is making that choice. We might as well think then of society as   an {\em individual} facing a choice problem, an analogy which invites us to ask: Does society have well-defined preferences over the alternative states? If this is so, then we are well-justified to think of society as an individual, on a purely behavioral basis, even if society does not have an independent existence and her preferences cannot be the result of {\em introspection}. We propose a {\em theory of noncooperative games} that provides support for the validity of this thought experiment as it implies that {\em it is possible to represent the interaction of multiple individuals as one preference relation that characterizes their equilibrium behavior}.

Mathematically, the theory we propose is akin to Nonlinear Programming whose fundamental problem is how an {\em individual} decision maker (instead of a society of such individuals) can optimally decide among a set of alternatives. Interestingly, both Nonlinear Programming and Noncooperative Game Theory were developed in parallel. In this paper, we establish a link between them.

\subsection{Nonlinear Programming}

Bertsekas defines Nonlinear Programming as the problem of minimizing a nonlinear function $f: X \rightarrow \mathbb{R}$  over a constraint set $X \subseteq \mathbb{R}^n$ that is specified by nonlinear equations and inequalities~\cite{Bertsekas}. The {\em solution concepts} in Nonlinear Programming are the global and local {\em minima}. A {\em global minimum} is a vector $x^* \in X$ such that $$f(x^*) \leq f(x), \mbox{\hspace{5mm}} \forall x \in X.$$ Unless $f$ is convex, there do not generally exist efficient algorithms for computing {\em global} minima. Instead nonlinear programming algorithms (such as {\em gradient descent}) strive to compute local minima. A {\em local minimum} is a vector $x^* \in X$ such that $$f(x^*) \leq f(x), \mbox{\hspace{5mm}} \forall x \in O,$$ where $O$ is a neighborhood of $x^*$. Of central importance in Nonlinear Programming are the first-order necessary conditions for optimality, namely, the {\em Karush-Kuhn-Tucker (KKT) conditions}~\cite{KT}.


\subsection{Noncooperative Game Theory}

Myerson defines Game Theory as ``[T]he study of mathematical models of conflict and cooperation between intelligent rational decision-makers''~\cite{Myerson}. The fundamental object of study of Noncooperative Game Theory is the normal-form game, which is a triple $(I, (S_i)_{i \in I}, (u_i)_{i \in I}),$ where $I$ is the set of players, $S_i$ is the set of pure strategies available to player $i$, and $u_i: S \rightarrow \mathbb{R}$ is the utility function of player $i$ where $S = \times_i S_i$ is the set of all strategy profiles (combinations of strategies). 

In this paper, we work with {\em population games}, which although being the fundamental object of study of Evolutionary Game Theory, in fact, {\em generalize} normal-form games. Furthermore, instead of {\em utility functions}, we are going to work with {\em cost functions}.

A population game (for example, see~\cite{PopulationGames}) $\mathcal{G}$ is a pair  $(X, c)$. $X$, the game's {\em state space} or {\em strategy profile space}, has product form, i.e., $X = X_1 \times \cdots \times X_n$, where the $X_i$'s are simplexes and $i=1,\ldots,n$ refers to a {\em player position} (or {\em population}). To each player position corresponds a set $S_i$ of $m_i$ {\em pure strategies} and a {\em mass} $\omega_i > 0$ (the population mass). The {\em strategy space} $X_i$ of player position $i$ has the form
\begin{align*}
X_i = \left\{ x \in \mathbb{R}^{m_i}  \bigg| \sum_{j \in S_i} x_j  = \omega_i \right\}.
\end{align*}
We refer to the elements of $X$ as {\em states} or {\em strategy profiles}. Each strategy profile $x \in X$ can be decomposed into a vector strategies, i.e., $x = (x_1,\ldots,x_n)'$. Let $m = \sum_i m_i$. $c: X \rightarrow \mathbb{R}^m$, the game's {\em cost function}, maps $X$, the game's state space, to vectors of {\em costs} where each position in the vector corresponds to a pure strategy. We assume that $c$ is {\em continuous}.

The fundamental solution concept in Noncooperative Game Theory is the {\em Nash equilibrium}~\cite{Nash}. However, the recent results of Daskalakis, Goldberg, and Papadimitriou on the hardness of computing Nash equilibria~\cite{Daskalakis} have cast doubt that the Nash equilibrium is a well-founded solution concept. The fundamental solution concept in Evolutionary Game Theory is the Evolutionarily Stable Strategy (ESS)~\cite{TheLogicOfAnimalConflict, Evolution}. Since its introduction by Maynard Smith and Price who defined the ESS in symmetric normal-form games, the definition of the ESS has been generalized to general population games. Unlike the Nash equilibrium, the complexity of computing an ESS is an open problem. In this paper, we show that the Nash equilibrium is for Noncooperative Game Theory what the KKT conditions are for Nonlinear Programming, namely, a {\em necessary optimality condition} and that being an ESS is sufficient for optimality. In what optimization problem though?

\subsection{Polytropic Programming}


Let $\mathcal{G} = (X, c)$ be a population game. Let $\preceq_\mathcal{G}$ be a preference relation on $X$ such that for $x,y \in X$, $$x \preceq_\mathcal{G} y \Leftrightarrow x \cdot c(y) \leq y \cdot c(y).$$ If $x \preceq_\mathcal{G} y$, we say that $x$ {\em weakly invades} $y$. $\preceq_\mathcal{G}$ is neither complete nor transitive; in fact, it is in general not even {\em acyclic}. However, $\prec_\mathcal{G}$ is strictly antisymmetric, and, therefore, $\preceq_\mathcal{G}$ is indeed a preference relation. 

Let $x^* \in X$, let $O \subseteq X$ be a neighborhood of $x^*$  and consider the following definitions:
\begin{definition}
$x^*$ is a Nash equilibrium of $\mathcal{G}$ if $\forall x \in X: x \cdot c(x^*) \geq x^* \cdot c(x^*)$.
\end{definition}
\begin{definition}
$x^*$ is a local minimum of $\preceq_\mathcal{G}$ if if $\exists O \subseteq X$ such that  $\forall x \in O: x^* \preceq_\mathcal{G} x$.
\end{definition}
\begin{definition}
$x^*$ is a minimal element of $\preceq_\mathcal{G}$ if $\forall x \in X: \neg (x \prec_\mathcal{G} x^*)$.
\end{definition}
\begin{definition}
$x^*$ is a strict local minimum of $\preceq_\mathcal{G}$ if $\exists O \subseteq X$ such that $\forall x \in O: x^* \prec_\mathcal{G} x$.
\end{definition}
\begin{definition}
$x^*$ is an ESS of $\mathcal{G}$ if if $\exists O \subseteq X$ such that $\forall x \in O: x^* \cdot c(x) < x \cdot c(x)$.
\end{definition}

Furthermore, let
\begin{align*}
\mbox{NE}(\mathcal{G}) &= \{ x^* | x^* \mbox{ is a Nash equilibrium of } \mathcal{G} \}\\
\mbox{LM}(\mathcal{G}) &= \{ x^* | x^* \mbox{ is a local minimum of } \preceq_\mathcal{G} \}\\
\mbox{PLP}(\mathcal{G}) &= \{ x^* | x^* \mbox{ is a minimal element of } \preceq_\mathcal{G} \}\\
\mbox{SLM}(\mathcal{G}) &= \{ x^* | x^* \mbox{ is a strict local minimum of } \preceq_\mathcal{G} \}\\
\mbox{ESS}(\mathcal{G}) &= \{ x^* | x^* \mbox{ is an ESS of } \mathcal{G} \}.
\end{align*}
We show that $$\mbox{ESS}(\mathcal{G}) = \mbox{SLM}(\mathcal{G}) \subseteq \mbox{PLP}(\mathcal{G}) \subseteq \mbox{LM}(\mathcal{G}) \subseteq \mbox{NE}(\mathcal{G}).$$
By virtue of this result, we call the minimal elements of $\preceq_\mathcal{G}$ {\em minimal Nash equilibria}. Note that ESS's are also minimal Nash equilibria, however, the latter are more general than the former. We call the problem of computing an element of $\mbox{PLP}(\mathcal{G})$, {\em Polytropic Programming}.

\subsubsection*{Kuhn and Tucker meet Nash}

Note further that being a Nash equilibrium is for Polytropic Programming what the KKT conditions are for Nonlinear Programming, namely, a necessary optimality condition. In fact, Sandholm has shown that if $\mathcal{G}$ is a full potential game (i.e., if there exists a smooth $f: X \rightarrow \mathbb{R}$ such that $c(\cdot) = \nabla f(\cdot)$), then $\mathcal{G}$'s Nash equilibria are precisely the states that satisfy the KKT conditions~\cite{xxx}. We further show that {\em in full potential games ESS's are precisely the strict local minima of} $f$, and, therefore, Polytropic Programming generalizes Nonlinear Programming.

\subsection{Monotropic games and multiplicative updates}

In the second part of the paper, we study games where being Nash is not only necessary for optimality but also sufficient. Consider the following definition:
\begin{definition}
$x^*$ is a minimum of $\preceq_\mathcal{G}$ if $\forall x \in X: x^* \preceq_\mathcal{G} x$.
\end{definition}
Let
\begin{align*}
\mbox{MNP}(\mathcal{G}) &= \{ x^* | x^* \mbox{ is a minimum of } \preceq_\mathcal{G} \}.
\end{align*}
We show that if $\mathcal{G}$ has a minimum, then $\mbox{PLP}(\mathcal{G}) = \mbox{MNP}(\mathcal{G}) = \mbox{LM}(\mathcal{G}) = \mbox{NE}(\mathcal{G})$, and if this happens we say that $\mathcal{G}$ is {\em monotropic} and call the problem of finding such a minimum {\em Monotropic Programming}. Many interesting games are monotropic such as {\em prisoner's dilemma}, {\em zero-sum games}, and {\em selfish routing games}. The latter games are members of a larger class of games, called {\em stable games}~\cite{xxx} that generalize potential games with a convex potential function, and stable games are easily shown to be monotropic. In their study of stable games, Hofbauer and Sandholm~\cite{xxx} define the concepts of Globally Neutrally Stable States (GNSS's) and Globally Evolutionarily Stable States (GESS's), and show that the equilibria of stable games fall into the class of such equilibrium concepts. In this paper, we show that these equilibrium concepts, in fact, {\em characterize monotropic games} in that a game is monotropic if and only if its equilibrium is a GNSS or GESS.

What's more, the equilibria of monotropic games are {\em easy to compute:} We show that these equilibria are {\em interior globally attractive} (in a dynamical systems sense) under {\em multiplicative updates}~\cite{xxx, yyy}. Although Hofbauer and Sandholm~\cite{xxx} show that the equilibria of stable games are attractive under various continuous-time dynamics including the {\em replicator dynamic}~\cite{xxx} that is the continuous analog of multiplicative updates, discrete-time analogs of continuous evolution rules often exhibit different behavior than their continuous-time counterparts. In this paper, we show that this is not the case for the replicator dynamic, however, our proof is technically more involved and a non-trivial application of the techniques in the continuous-time case. Since monotropic games generalize zero-sum games, our result on the global stability of the equilibria of monotropic games under multiplicative updates generalizes the earlier result of Freund and Schapire~\cite{xxx} who show that minimax equilibria have the same property.

\subsection{A stability theory of noncooperative games}

Not every game has a minimal Nash equilibrium; for example, the {\em rock-paper-scissors game} doesn't~\cite{xxx}, although {\em acyclic games} always do~\cite{xxx}. Therefore, minimal equilibria cannot be Noncooperative Game Theory's general solution concept. One of the main contributions of this paper is that it provides a conceptual framework toward devising such a solution concept.

\subsubsection{Minimal Nash equilibria are attractive}

Nash equilibria are generally understood as {\em fixed points} of the game's best-response correspondence. The equilibria of a {\em dynamical system} are also fixed points of a map, however, dynamical systems are endowed with a rich {\em stability theory} that noncooperative games lack. 

One of the most universal concepts in Dynamical Systems Theory is that of {\em Lyapunov stability:} An equilibrium $x^*$ is called Lyapunov stable if all solutions of the dynamical system that start near $x^*$ converge to $x^*$. In this theory, Lyapunov stability is a {\em refinement} of the concept of equilibrium in the sense that being an equilibrium is a necessary condition for Lyapunov stability, but it is not sufficient. 

Our final result is that the set $X^*$ of minimal Nash equilibria is {\em setwise attractive} in the sense that there always exists a neighborhood $O \subseteq X$ of this set such that for every $x \in O - X^*$, there exists an $x^* \in X^*$ such that $x^* \prec_{\mathcal{G}} x$. By virtue of this result, we also call minimal Nash equilibria {\em attractive equilibria}.

The fundamental reason why Lyapunov stability is an interesting concept is also the fundamental reason why attractive equilibria are a more natural (albeit not general) solution concept than Nash equilibria: {\em If it is not possible to converge to an equilibrium, it is unlikely that this equilibrium will emerge in a physical system.} Note that noncooperative games are {\em not} dynamical systems, as they do not have a natural evolution rule, however, the preference relation we have introduced is a natural criterion to base a stability theory on.

\subsubsection{Toward a general solution concept for Noncooperative Game Theory}

In that sense, we believe that this paper is a first step toward devising a {\em general stability theory of noncooperative games} whose ultimate goal would be to discover Noncooperative Game Theory's general solution concept, which has to be more general than a fixed point: {\em In all likelihood Noncooperative Game Theory's natural solution concept is a general attractor.} What we have shown in this paper is that attractive Nash equilibria are a special case of this more general solution concept.

\subsection{Related work}

\subsubsection{Theory of binary relations}

Polytropic Programming is a problem of optimizing a preference relation. Binary relations are the fundamental object of study of Order Theory (for example, see~\cite{xxx}), however, this theory is primarily concerned with partially ordered sets whereas the preference relation of a polytropic program is not generally transitive. Several papers study {\em acyclic relations}~\cite{xxx}, however, a polytropic program may not even be acyclic in general, and our theory has been developed without this assumption. Nevertheless, we expect the concept of acyclicity to assume a central role in the theory of polytropic programs as this theory is further developed.

\subsubsection{Algorithmic Game Theory}

Algorithmic Game Theory~\cite{xxx, yyy} is the study of Game Theory under the ``algorithmic lens.'' Perhaps the most significant result of this theory to this day is that the problem of computing a Nash equilibrium is PPAD-Complete~\cite{xxx}. The complexity of computing an attractive Nash equilibrium or determining whether such an equilibrium exists are interesting open problems. Etessami and Lochbihler~\cite{xxx} show that the problem of determining whether an ESS exists is both NP-Hard and coNP-Hard, and these problems are generally believed to be polynomial.

\subsubsection{Evolutionary Game Theory}

According to Wikipedia, ``Evolutionary game theory (EGT) is the application of game theory to evolving populations of lifeforms in biology.''\footnote{\url{http://en.wikipedia.org/wiki/Evolutionary_game_theory}} The origins of EGT can be traced to the introduction of the ESS by Maynard Smith and Price~\cite{xxx} (see also~\cite{yyy}), and since then EGT has been a flourishing field. The fundamental object of study of EGT is the {\em population game} (for example, see~\cite{xxx}), which is generally not understood as a generalization of the normal-form game, although it is. This paper by and large argues that many answers to Game Theory's impasse can be found in {\em evolution}, which raises an important question on the nature of {\em rationality}.

\subsubsection{Theory of learning in games and evolutionary dynamics}

Freund and Schapire~\cite{FreundSchapire2} show that Hedge (called {\em MW} in that paper for {\em multiplicative weights}) can approximately solve a zero-sum game and, therefore, any linear program. The solution of linear and semidefinite programs using multiplicative weights is the subject of Kale's PhD thesis (see \cite{Kale} and references therein). In this paper, we offer a {\em nonlinear programming perspective} on Hedge, and, in fact, more as we show that it solves {\em polytropic programs}, which generalize nonlinear programs. 

A class of games that has received particular attention is the class of atomic and nonatomic {\em congestion games}, which are potential games~\cite{xxx, yyy}. Hedge has found applications in the combinatorial optimization problem of computing equilibria in atomic congestion games. Kleiberg, Piliouras, and Tardos \cite{Piliouras} show that if the players in an atomic congestion game use Hedge when they probabilistically choose their strategies, then, for a small enough learning rate, the dynamics of the system can be approximated by the continuous-time replicator dynamic and that Nash equilibria are attractive under this dynamic.

Nash equilibria in nonatomic congestion games (which are population games) are also called Wardrop equilibria due to the early studies of Wardrop in road-traffic transportation systems~\cite{xxx}. Several researchers have proposed simple decentralized learning rules that if followed by individual agents in a population the resulting aggregate result is a dynamic that converges to Wardrop equilibria. Fischer and Voecking~\cite{SelfishRoutingEvolution} show that the (continuous-time) replicator dynamic converges to Wardrop equilibria. \cite{FastConvergence-Journal} consider adaptive rerouting policies in a round-based (discrete-time) model and show convergence to Wardrop equilibria such that, in symmetric games with polynomial latency functions, the number of rounds is polynomial in the representation length of the latency functions. \cite{RoutingRegret} show convergence to Wardrop equilibria of discrete-time evolution rules in which the agents in a population continuum use no-regret learning. 

\cite{StableGames} consider {\em stable games}, which include non-atomic congestion games as a special case, and show convergence of a wide range of dynamics (based on continuous-time revision protocols using a Poisson clock) to their corresponding Nash equilibria.

\section{Preliminaries}

\subsection{Preference relations}

Let $X$ be a nonempty set. A subset $R$ of $X \times X$ is called a {\em binary relation} on X. If $(x,y) \in R$, we write $xRy$, and if $(x,y) \not\in R$, we write $\neg xRy$. $R$ is called {\em symmetric} if, for any $x, y \in X$, we have that $xRy \Rightarrow yRx$. $R$ is called {\em antisymmetric} if, for any $x, y \in X$, $xRy \wedge yRx \Rightarrow x=y$. 

Let $xPy \Leftrightarrow xRy \wedge \neg yRx$ and $xIy \Leftrightarrow xRy \wedge yRx$. Then $P$ and $I$ are also binary relations on X where $P \subset R$ and $I \subset R$. $P$ is called the asymmetric part of $R$ and $I$ is called its symmetric part. If $P$ is antisymmetric, then $R$ is called a {\em preference relation}.

Let $\succeq$ be a preference relation on $X$, and let $Y$ be a nonempty subset of $X$. An element $x$ of $Y$ is said to be $\succeq$-maximal in $Y$ if there is no $y \in Y$ with $y \succ x$.

\subsection{Invasion, dominance, and equilibria}

Let $(X, c)$ be a population game, let $x, y \in X$, and let $x_\epsilon = (1-\epsilon) x + \epsilon y$, $\epsilon \in [0,1]$.

\begin{definition}
\label{xcnbvlzfjv}
We say that $x$ weakly invades $y$ if $x \cdot c(y) \leq y \cdot c(y)$. If the inequality is strict, we say that $x$ invades $y$. If $\forall \epsilon \in [0,1], x \cdot c(x_\epsilon) \leq y \cdot c(x_\epsilon)$, we say that $x$ weakly dominates $y$. If the inequality is strict, we say that $x$ dominates $y$. 
\end{definition}

\begin{definition}
\label{xcnbvlzfjvx}
Let $x^* \in X$. Then $x^*$ is a Nash equilibrium if for all $x \in X$, $x \cdot c(x^*) \geq x^* \cdot c(x^*)$.
\end{definition}

\begin{definition}
\label{xcnbvlzfjvx}
Let $x^* \in X$. Then $x^*$ is a Neutrally Stable State (NSS) if for all $x \in O$, $x^* \cdot c(x) \leq x \cdot c(x)$ where $O$ is a neighborhood of $x^*$. If the inequality is strict, $x^*$ is an Evolutionarily Stable State (ESS).
\end{definition}

\subsection{Preliminary results}

$...$

\section{Polytropic Programming}

Let $\mathcal{G}=(X,c)$ be a population game, $x,y \in X$, $x_\epsilon = (1-\epsilon) x + \epsilon y$, $\epsilon \in [0,1]$, and $\delta(\epsilon) = (x-y) \cdot c(x_\epsilon)$.

\begin{lemma}
\label{qwerxmcbpdkfjhbdfg}
If $x$ dominates $y$, then $x \prec_\mathcal{G} y$.
\end{lemma}

\begin{lemma}
\label{slakfdjdskfj}
If $x$ invades $y$, then there exists $0 \leq \epsilon < 1$ such that $x_\epsilon$ that dominates $y$.
\end{lemma}

\begin{proof}
By assumption, $\delta(1) < 0$, and, therefore, there exists $0 \leq \epsilon < 1$ such that, for all $ \epsilon \leq \epsilon' \leq 1$, $\delta(\epsilon') < 0$, which completes the proof.
\end{proof}

\begin{lemma}
\label{weporitkjfhgdfpjgn}
If, for all $0 \leq \epsilon < 1$, $x_\epsilon$ does not dominate $y$, then $x$ does not invade $y$.
\end{lemma}

\begin{proof}
If $x$ invades $y$, then by Lemma~\ref{slakfdjdskfj}, there exists such $x_\epsilon$ that dominates $y$, a contradiction.
\end{proof}

\begin{theorem}
$\mbox{{\em PLP}}(\mathcal{G}) \subseteq \mbox{{\em NE}}(\mathcal{G})$.
\end{theorem}

\begin{proof}
Let $x^* \in \mbox{PLP}(\mathcal{G})$. Then, $\forall x \in X, \mbox{ } \neg (x \prec x^*)$, and, by Lemma~\ref{qwerxmcbpdkfjhbdfg}, $\forall x \in X, \mbox{ }  x$ does not dominate $ x^*$. For any $x \in X-\{x^*\}$, $x_\epsilon$ does not dominate $x^*$, and, therefore, by Lemma~\ref{weporitkjfhgdfpjgn}, $x$ does not invade $x^*$, which implies that $x^*$ is a Nash equilibrium.
\end{proof}

\begin{theorem}
$\mbox{{\em PLP}}(\mathcal{G}) \subseteq \mbox{{\em PLP$'$}}(\mathcal{G})$. 
\end{theorem}

\begin{proof}
Let $x^*$ be $\preceq_{\mathcal{G}}$-minimal. It suffices to show that, for all $x \in X$, there exists $\epsilon' > 0$ such that, for all $\epsilon \in [0, \epsilon']$, $x^* \cdot c(x_\epsilon) \leq x_\epsilon \cdot c(x_\epsilon)$ where $x_\epsilon = (1-\epsilon) x^* + \epsilon x$. By the minimality of $x^*$, we have
\begin{align*}
\forall x, \neg (x \prec_\mathcal{G} x^*) &\Leftrightarrow \forall x, x \cdot c(x^*) > x^* \cdot c(x^*) \vee x^* \cdot c(x) \leq x \cdot c(x)
\end{align*}
Now let $x \in X$. If $x^* \cdot c(x_\epsilon) \leq x_\epsilon \cdot c(x_\epsilon)$, then ... (done) Suppose now that $x^* \cdot c(x) > x \cdot c(x)$. Then $x \cdot c(x^*) > x^* \cdot c(x^*)$, which implies that ... (done)
\end{proof}

\begin{theorem}
$\mbox{{\em PLP$'$}}(\mathcal{G}) \subseteq \mbox{{\em NE}}(\mathcal{G})$. 
\end{theorem}

\begin{proof}
Let $x^* \in \mbox{PLP$'$}(\mathcal{G})$. Then there exists $O \subseteq X$ such that, for all $x \in O$, $x^* \cdot c(x) \leq x \cdot c(x)$. Suppose there exists $y \in X$ such that $y \cdot c(x^*) < x^* \cdot c(x^*)$. Then, by Lemma ..., there exists $y_\epsilon$ that dominates $x^*$. Let $y_\epsilon \in O$. Then $x^* \cdot c(y_\epsilon) \leq y_\epsilon \cdot c(y_\epsilon)$
\end{proof}

Some thoughts:

$$x \prec y \Leftrightarrow x \cdot c(y) \leq y \cdot c(y) \wedge y \cdot c(x) > x \cdot c(x)$$

$$\neg (y \prec x) \Leftrightarrow y \cdot c(x) > x \cdot c(x) \vee x \cdot c(y) \leq y \cdot c(y)$$\\

Suppose that $\neg (y \prec x)$ and $x \cdot c(y) > y \cdot c(y)$. Then $y \cdot c(x) > x \cdot c(x)$.

\begin{lemma}
If $\neg (y \prec x)$ and $\neg (x \preceq y)$, then $\neg (y \preceq x)$.
\end{lemma}

Suppose that $\neg (y \prec x)$ and $y \cdot c(x) \leq x \cdot c(x)$. Then $x \cdot c(y) \leq y \cdot c(y)$.

\begin{lemma}
If $\neg (y \prec x)$ and $y \preceq x$, then $x \preceq y$.
\end{lemma}

Suppose that, $\forall y$, $\neg (y \prec x)$ and $y \cdot c(x) \geq x \cdot c(x)$. Then

$$\forall y, (y \cdot c(x) > x \cdot c(x)) \vee ((y \cdot c(x) = x \cdot c(x)) \wedge (x \cdot c(y) \leq y \cdot c(y)))$$

-----

Suppose that, $\forall y$, $\neg (y \prec x)$ and $\neg (y \prec x')$. So we have

$$\neg (x \prec x') \Leftrightarrow x \cdot c(x') > x' \cdot c(x') \vee x' \cdot c(x) \leq x \cdot c(x)$$

$$\neg (x' \prec x) \Leftrightarrow x' \cdot c(x) > x \cdot c(x) \vee x \cdot c(x') \leq x' \cdot c(x')$$

Both $x$ and $x'$ are Nash equilibria (we have shown that). So we have

$$x \cdot c(x') \geq x' \cdot c(x')$$

$$x' \cdot c(x) \geq x \cdot c(x)$$

Suppose that $x' \cdot c(x) > x \cdot c(x)$. Then $x \cdot c(x') > x' \cdot c(x')$. Conversely, suppose that  $x \cdot c(x') > x' \cdot c(x')$. Then $x' \cdot c(x) > x \cdot c(x)$. Good!

Suppose that $x' \cdot c(x) = x \cdot c(x)$. Then $x \cdot c(x') \leq x' \cdot c(x')$. However, $x \cdot c(x') \geq x' \cdot c(x')$. Therefore, $x \cdot c(x') = x' \cdot c(x')$. The converse is also true. Good!

-----

Suppose that, $\forall y$, $\neg (y \prec x)$ and $\neg (y \prec x')$. So we have

$$\neg (x_\epsilon \prec x') \Leftrightarrow x_\epsilon \cdot c(x') > x' \cdot c(x') \vee x' \cdot c(x_\epsilon) \leq x_\epsilon \cdot c(x_\epsilon)$$

$$\neg (x_\epsilon \prec x) \Leftrightarrow x_\epsilon \cdot c(x) > x \cdot c(x) \vee x \cdot c(x_\epsilon) \leq x_\epsilon \cdot c(x_\epsilon)$$

Both $x$ and $x'$ are Nash equilibria (we have shown that). So we have

$$x_\epsilon \cdot c(x') \geq x' \cdot c(x')$$

$$x_\epsilon \cdot c(x) \geq x \cdot c(x)$$

We have

$$\neg (x_\epsilon \prec x') \Leftrightarrow x \cdot c(x') > x' \cdot c(x') \vee x' \cdot c(x_\epsilon) \leq x \cdot c(x_\epsilon)$$

$$\neg (x_\epsilon \prec x) \Leftrightarrow x' \cdot c(x) > x \cdot c(x) \vee x \cdot c(x_\epsilon) \leq x' \cdot c(x_\epsilon)$$

And

$$x \cdot c(x') \geq x' \cdot c(x')$$

$$x' \cdot c(x) \geq x \cdot c(x)$$

Suppose that $x' \cdot c(x) > x \cdot c(x)$. Then $x \cdot c(x') > x' \cdot c(x')$. Conversely, suppose that  $x \cdot c(x') > x' \cdot c(x')$. Then $x' \cdot c(x) > x \cdot c(x)$. Good!

Suppose that $x' \cdot c(x) = x \cdot c(x)$. Then $x \cdot c(x_\epsilon) \leq x' \cdot c(x_\epsilon)$. However, $x \cdot c(x') \geq x' \cdot c(x')$. Therefore, $x \cdot c(x') = x' \cdot c(x')$. And, thus, $x \cdot c(x_\epsilon) = x' \cdot c(x_\epsilon)$. 

-----

This shows that the equivalence relation is symmetric. Now I need to show that it is transitive. Let $x''$ be a minimal Nash equilibrium and suppose that $x'' \sim x'$ and $x' \sim x$. Then I will show that $x'' \sim x$. So we have:

$$x' \cdot c(x'_\epsilon) = x'' \cdot c(x'_\epsilon)$$

$$x \cdot c(x_\epsilon) = x' \cdot c(x_\epsilon)$$\\

We have $x \cdot c(x) = x' \cdot c(x)$ and $x \cdot c(x') = x' \cdot c(x')$.

Furthermore, $x' \cdot c(x') = x'' \cdot c(x')$ and $x' \cdot c(x'') = x'' \cdot c(x'')$

-----

Suppose $x \cdot c(x') \leq x' \cdot c(x')$ and $x' \cdot c(x'') \leq x'' \cdot c(x'')$. I would like to show that $x \cdot c(x'') \leq x'' \cdot c(x'')$.

Because $x''$ is Nash we have $x \cdot c(x'') \geq x'' \cdot c(x'')$

-----

Suppose that $x$ and $x'$ are minimal equilibria and suppose further that $x' \cdot c(x) = x \cdot c(x)$. Then I would like to show that they are {\em equivalent} in the sense that $x \preceq y \Leftrightarrow x' \preceq y$. We need to consider cases. Suppose that $y$ is also a minimal Nash equilibrium.

$x' \cdot c(x_\epsilon)$\\

$x \cdot c(y) \leq y \cdot c(y) \Rightarrow x_\epsilon \cdot c(y) \leq y \cdot c(y)$

$x \cdot c(x_\epsilon) \leq x_\epsilon \cdot c(x_\epsilon) \Rightarrow x \cdot c(x_\epsilon) \leq y \cdot c(x_\epsilon)$

-----

Suppose that, $\forall y$, $\neg (y \prec x^*)$ and $x^* \prec x$. So we have

$$\neg (x_\epsilon \prec x^*) \Leftrightarrow x \cdot c(x^*) > x^* \cdot c(x^*) \vee x^* \cdot c(x_\epsilon) \leq x \cdot c(x_\epsilon)$$

$$x^* \prec x \Leftrightarrow x^* \cdot c(x) \leq x \cdot c(x) \wedge x \cdot c(x^*) > x^* \cdot c(x^*)$$

-----

$$x \cdot A x_\epsilon = x' \cdot A x_\epsilon$$

$$x \cdot A ((1-\epsilon)x + \epsilon x') = x' \cdot A ((1-\epsilon)x + \epsilon x')$$

$$x \cdot A x + \epsilon (x \cdot A x' - x \cdot A x) = x' \cdot A x + \epsilon (x' \cdot A x' - x' \cdot A x)$$

-----

$$x \cdot c(x^*) \geq x^* \cdot c(x^*)$$

$$\neg (x \prec x^*) \Leftrightarrow x \cdot c(x^*) > x^* \cdot c(x^*) \vee x^* \cdot c(x) \leq x \cdot c(x)$$

-----

$$\neg (x \vee y) \vee (x \wedge y)$$

$$(\neg x \wedge \neg y) \vee (x \wedge y)$$

-----

$$\neg x \vee y$$

-----

$$\exists y: y \prec x \Rightarrow \exists z: z \cdot c(x) < x \cdot c(x)$$

-----

Suppose that for all $x, y$

$$x \cdot c(y) \geq y \cdot c(y) \Leftrightarrow \neg (x \prec y)$$

$$(x \cdot c(y) < y \cdot c(y) \wedge x \prec y) \vee (x \cdot c(y) \geq y \cdot c(y) \wedge \neg (x \prec y))$$

$$(x \cdot c(y) < y \cdot c(y)) \vee (x \cdot c(y) \geq y \cdot c(y) \wedge (x \cdot c(y) > y \cdot c(y) \vee y \cdot c(x) \leq x \cdot c(x)))$$

$$(x \cdot c(y) < y \cdot c(y)) \vee (x \cdot c(y) > y \cdot c(y)) \vee (x \cdot c(y) \geq y \cdot c(y) \wedge y \cdot c(x) \leq x \cdot c(x)))$$

-----

$$\forall x \in O, x^* \preceq x$$

$$\neg (x \prec x^*) \Leftrightarrow x \cdot c(x^*) > x^* \cdot c(x^*) \vee x^* \cdot c(x) \leq x \cdot c(x)$$

Take arbitrary $y$ (say outside of $O$). Then

$$\neg (y \prec x^*) \Leftrightarrow y \cdot c(x^*) > x^* \cdot c(x^*) \vee x^* \cdot c(y) \leq y \cdot c(y)$$

-----

Suppose that 

$$\forall x: x \cdot c(x^*) \geq x^* \cdot c(x^*) \Leftrightarrow \forall x: \neg (x \prec x^*)$$

$$(\exists y: y \cdot c(x^*) < x^* \cdot c(x^*)) \vee (\forall x: x \cdot c(x^*) \geq x^* \cdot c(x^*) \wedge\neg (x \prec x^*))$$

$$(\exists y: y \cdot c(x^*) < x^* \cdot c(x^*)) \vee (\forall x: x \cdot c(x^*) \geq x^* \cdot c(x^*) \wedge (x \cdot c(x^*) > x^* \cdot c(x^*) \vee x^* \cdot c(x) \leq x \cdot c(x)))$$

$$(\exists y: y \cdot c(x^*) < x^* \cdot c(x^*)) \vee (\forall x: x \cdot c(x^*) > x^* \cdot c(x^*)) \vee (\forall x: x \cdot c(x^*) \geq x^* \cdot c(x^*) \wedge x^* \cdot c(x) \leq x \cdot c(x)))$$\\

Suppose that 

$$\forall x: x \cdot c(x_1^*) > x_1^* \cdot c(x_1^*) \mbox{ and } \forall x: x \cdot c(x_2^*) > x_2^* \cdot c(x_2^*)$$\\

Take $x^*_\epsilon = (1-\epsilon) x_1^* + \epsilon x_2^*$. Then $(x^*_2 - x^*_1) \cdot c(x^*_\epsilon)$

\section{Monotropic games}

Consider the restriction of $\preceq_\mathcal{G}$ on $X^*$, the set of minimal equilibria of $\mathcal{G}$. First we show that $(X^*, \preceq_\mathcal{G})$ is {\em reflexive} and {\em symmetric} (however, note that, in general, it may not be transitive).

\begin{lemma}
Let $x, x' \in X^*$ and let $x_\epsilon = (1-\epsilon) x + \epsilon x'$ where $\epsilon \in [0,1]$. Then 
\begin{description}
\item[{\em (i)}] $x' \cdot c(x) > x \cdot c(x) \Leftrightarrow x \cdot c(x') > x' \cdot c(x')$.
\item[{\em (ii)}] $x' \cdot c(x) = x \cdot c(x) \Rightarrow \forall \epsilon: x' \cdot c(x_\epsilon) = x \cdot c(x_\epsilon)$.
\item[{\em (ii)}] $x \cdot c(x') = x' \cdot c(x') \Rightarrow \forall \epsilon: x \cdot c(x_\epsilon) = x' \cdot c(x_\epsilon)$.
\end{description}
\end{lemma}

\begin{proof}
Suppose that, $\forall y$, $\neg (y \prec x)$ and $\neg (y \prec x')$. So we have

$$\neg (x_\epsilon \prec x') \Leftrightarrow x_\epsilon \cdot c(x') > x' \cdot c(x') \vee x' \cdot c(x_\epsilon) \leq x_\epsilon \cdot c(x_\epsilon)$$

$$\neg (x_\epsilon \prec x) \Leftrightarrow x_\epsilon \cdot c(x) > x \cdot c(x) \vee x \cdot c(x_\epsilon) \leq x_\epsilon \cdot c(x_\epsilon)$$

Both $x$ and $x'$ are Nash equilibria (we have shown that). So we have

$$x_\epsilon \cdot c(x') \geq x' \cdot c(x')$$

$$x_\epsilon \cdot c(x) \geq x \cdot c(x)$$

We have

$$\neg (x_\epsilon \prec x') \Leftrightarrow x \cdot c(x') > x' \cdot c(x') \vee x' \cdot c(x_\epsilon) \leq x \cdot c(x_\epsilon)$$

$$\neg (x_\epsilon \prec x) \Leftrightarrow x' \cdot c(x) > x \cdot c(x) \vee x \cdot c(x_\epsilon) \leq x' \cdot c(x_\epsilon)$$

And

$$x \cdot c(x') \geq x' \cdot c(x')$$

$$x' \cdot c(x) \geq x \cdot c(x)$$

Suppose that $x' \cdot c(x) > x \cdot c(x)$. Then $x \cdot c(x') > x' \cdot c(x')$. Conversely, suppose that  $x \cdot c(x') > x' \cdot c(x')$. Then $x' \cdot c(x) > x \cdot c(x)$. Good!

Suppose that $x' \cdot c(x) = x \cdot c(x)$. Then $x \cdot c(x_\epsilon) \leq x' \cdot c(x_\epsilon)$. However, $x \cdot c(x') \geq x' \cdot c(x')$. Therefore, $x \cdot c(x') = x' \cdot c(x')$. And, thus, $x \cdot c(x_\epsilon) = x' \cdot c(x_\epsilon)$. 
\end{proof}

\begin{lemma}
$\preceq^*_\mathcal{G}$ is {\em reflexive} and {\em symmetric}.
\end{lemma}

\begin{proof}
$...$
\end{proof}

\begin{theorem}
If $\mathcal{G} = (X, c)$ is monotropic, then $\preceq^*_\mathcal{G}$ is an equivalence relation.
\end{theorem}

Let $\mathcal{G} = (X, c)$ be monotropic and let $x^*$ be a Nash equilibrium of $\mathcal{G}$. Since $x^*$ is a Nash equilibrium we have $$x \cdot c(x^*) \geq x^* \cdot c(x^*).$$ Suppose we have shown that minima are also Nash equilibria, then we also have $$x^* \cdot c(x) \leq x \cdot c(x).$$ 

$$x^* \cdot c(x_\epsilon) \leq x \cdot c(x_\epsilon)$$.

This implies 

$$x \cdot c(x^*) = x^* \cdot c(x^*).$$

Suppose that $x$ is also a minimum. Then 

$$x^* \cdot c(x_{\epsilon'}) \geq x \cdot c(x_{\epsilon'})$$

-----

$$(\neg x \wedge \neg y) \vee (x \wedge y)$$

Suppose that 

$$\forall x: x \cdot c(x^*) \geq x^* \cdot c(x^*) \Leftrightarrow \forall x: x^* \cdot c(x) \leq x \cdot c(x)$$

$$(\exists y: y \cdot c(x^*) < x^* \cdot c(x^*) \wedge \exists y: x^* \cdot c(y) > y \cdot c(y)) \vee (\forall x: x \cdot c(x^*) \geq x^* \cdot c(x^*) \wedge x^* \cdot c(x) \leq x \cdot c(x))$$

\section{Freund and Schapire meet Maynard Smith and Price}

\subsection{Multiplicative updates}

\subsubsection{Hedge}

Hedge is a map $T: X \rightarrow X$ defined by
\begin{align}
T_i(x) &= \left[ x_j^i \frac{ \exp\{ - \alpha c^i_j(x) \}}{(1 / \omega_i) \sum_{k \in P_i} x_k^i \exp\{- \alpha c^i_k(x)\}} \right]_j, \mbox{   } i=1,\ldots,n,\label{Hedge}
\end{align}
where we denote by $[ \cdot ]_j$ a vector whose $j$th element appears in the brackets.

\subsubsection{The discrete-time replicator dynamic}

The discrete-time replicator dynamic is a map $S: X \rightarrow X$ defined by
\begin{align}
S_i(x) &= \left[ x_j^i \frac{ 1- \alpha c^i_j(x) }{1- (\alpha / \omega_i) \sum_{k \in P_i} x_k^i c^i_k(x)} \right]_j\notag\\
&= \left[ x_j^i \frac{ 1- \alpha c^i_j(x) }{1- (\alpha / \omega_i) x^i \cdot c^i(x)} \right]_j\notag\\
&\equiv \left[ x_j^i g_j^i(x)\right]_j, \mbox{   } i=1,\ldots,n\label{DTRD}
\end{align}
This replicator equation compares the cost $c^i_j(x)$ of path $j$ with the average cost $(1 / \omega_i) x^i \cdot c^i(x)$. If $c^i_j(x) < (1 / \omega_i) x^i \cdot c^i(x)$, then $g^i_j(x) > 1$ and the flow on path $j$ is increased. If, however, $c^i_j(x) > (1 / \omega_i) x^i \cdot c^i(x)$, then $g^i_j(x) < 1$ and the flow on path $j$ is decreased. Note that the positive scalar $\alpha$ must be small enough so that $\max_i \max_j \max_{x \in X} (1 - \alpha c^i_j(x)) \geq 0$. Let $\bar{\alpha}$ denote this upper bound on $\alpha$.

\subsection{Fixed points of multiplicative updates}
\label{flsjakdflskdjflk}

The fixed points of Hedge are the points of $X$ in which, for each commodity, the costs of all paths that carry flow are equal.

\begin{theorem}
\label{kxcvjbsfojvasj}
Let $P_+(x^i) = \{ j \in P_i | x^i_j > 0 \}$. $x = (x_1, \ldots, x_n)'$ is a fixed point of Hedge iff for all $i$ and for all $j_1, j_2 \in P_+(x^i)$, $c^i_{j_1}(x) = c^i_{j_2}(x)$.
\end{theorem}

\begin{proof}
First we show sufficiency. That is, we show that
if for all $i$ and for all $j_1, j_2 \in P_+(x_i)$, $c^i_{j_1}(x) = c^i_{j_2}(x)$,
then $T(x) = x$.
Some of the coordinates of $x_i$
are zero and some are positive. Clearly, 
the zero coordinates will not become positive after applying the map. 
Now, notice that, for all $j \in P_+(x_i)$, $\exp\{- \alpha c^i_j(x)\} = \sum_{k \in P_i} x^i_k \exp\{- \alpha c^i_k(x)\}$, and, therefore, $T_i(x) = x^i$, and this is true for all $i$.

Now we show necessity. That is, we show that 
if $x$ is a fixed point of Hedge, 
then for all $i$ and for all $j_1, j_2 \in P_+(x_i)$, $c^i_{j_1}(x) = c^i_{j_2}(x)$.
Let $\hat{x}^i = T_i(x)$.
Because $x$ is a fixed point, $\hat{x}^i_j = x^i_j$. Therefore,
{\allowdisplaybreaks
\begin{align}
\hat{x}^i_j &= x^i_j\notag\\
\frac{x^i_j \exp{\{- \alpha c^i_j(x)\}}}{\sum_{k} x^i_k \exp{\{- \alpha c^i_k(x)\}}} &= x^i_j\notag\\
1 &= x^i_j \frac{\sum_{k} x^i_k \exp{\{- \alpha c^i_k(x)\}}}{x^i_j \exp{\{- \alpha c_j(x)\}}}\notag\\
1 &= \sum_k x^i_k \exp{\{- \alpha (c^i_k(x) - c^i_j(x))\}}\notag\\
1 &= x^i_j + \sum_{k \neq j} x^i_k \exp{\{- \alpha (c^i_k(x) - c^i_j(x))\}}\notag\\
1 &= 1 - \sum_{k \neq j} x^i_k + \sum_{k \neq j} x^i_k \exp{\{- \alpha (c^i_k(x) - c^i_j(x))\}}\notag\\
0 &= \sum_{k \neq j} x^i_k \left(\exp{\{- \alpha (c^i_k(x) - c^i_j(x))\}} - 1 \right)\label{eqcondition-finalstep}
\end{align}
}
Equation (\ref{eqcondition-finalstep}) implies that 
\[\exp{\{- \alpha (c^i_k(x) - c^i_j(x))\}} = 1, x^i_k > 0,\]
and, thus,
\[c^i_k(x) = c^i_j(x), x^i_k > 0.\]
\end{proof}

In a Wardrop equilibrium, in each commodity the cost of all paths that carry flow are equal and the cost of any path that does not carry flow is equal to or greater than the cost of the paths that carry flow. Therefore, the Wardrop equilibrium is a fixed point of Hedge. However, Hedge may have additional fixed points in which the cost of a path that does not carry flow is lower than the cost of paths that carry flow in the same commodity. Since these fixed points would satisfy the definition of a Wardrop equilibrium in a system with fewer paths, we call them {\em secondary Wardrop equilibria}.

It is possible to show by a technically detailed proof (which we omit for brevity) that the secondary Wardrop equilibria are {\em unstable} under Hedge. The intuition is simple: In a secondary Wardrop equilibrium at least one path without flow (say path $p$) has cost lower than the cost of the flow-carrying paths. What would Hedge do if we were to add a small amount of flow to $p$ and let Hedge drive the system? The answer is that Hedge would continue to add flow to $p$ in trying to balance the costs (given that the cost of $p$ is initially small), which implies that the secondary Wardrop equilibrium is not attractive and, therefore, unstable.

It can be shown by a proof analogous to that of Theorem~\ref{kxcvjbsfojvasj} that the fixed points of the discrete-time replicator dynamic and the fixed points of Hedge coincide.

\subsection{The replicator map as a form of gradient for $\succeq_G$}

Most interesting nonlinear optimization algorithms rely on the idea of {\em iterative descent}~\cite{xxx}, which is, starting at some point $x^0$, to generate a sequence $\{ x^k \}$ such that the objective function $f$ is decreased in each iteration, that is, $f(x^{k+1}) < f(x^k)$. An important class of iterative descent algorithms are the {\em gradient methods}, which rely on the fact that the gradient of $f$ is a {\em descent direction}. In an analogous manner, in this section, we show that the replicator map is a ``descent direction'' for $\succeq_G$. More precisely, we have the following lemma.

\begin{lemma}
Let $(X, c)$ be a population game, and let $x \in X$ such that $S(x) \neq x$. Then $S(x)$ invades $x$.
\end{lemma}

\begin{proof}
Let $x' = S(x)$. We have
\begin{align*}
RE(x', x) = \sum_{j \in P} x'_j \ln \frac{x'_j}{x_j} = \sum_{j \in P} x'_j \ln \frac{1-\alpha c_j(x)}{1-\alpha c(x|x)} > 0.
\end{align*}
Therefore,
\begin{align*}
\ln \left(1-\alpha c(x|x)\right) &< \sum_{j \in P} x'_j \ln \left( 1-\alpha c_j(x) \right)\\
1 - \alpha c(x|x) &< \exp \left\{ \sum_{j \in P} x'_j \ln \left( 1-\alpha c_j(x) \right) \right\}\\
c(x|x) &> \frac{1}{\alpha} \left( 1 - \exp \left\{ \sum_{j \in P} x'_j \ln \left( 1-\alpha c_j(x) \right) \right\} \right).
\end{align*}
Now
\begin{align*}
\exp \left\{ \sum_{j \in P} x'_j \ln \left( 1-\alpha c_j(x) \right) \right\} = \prod_{j \in P} \left( 1-\alpha c_j(x) \right)^{x'_j}
\end{align*}
and, therefore,
\begin{align*}
c(x|x) > \frac{1}{\alpha} \left(1 - \prod_{j \in P} \left( 1-\alpha c_j(x) \right)^{x'_j} \right).
\end{align*}
Furthermore,
\begin{align*}
\prod_{j \in P} \left( 1-\alpha c_j(x) \right)^{x'_j} \leq 1 - \alpha \sum_{j \in P} x'_j c_j(x),
\end{align*}
implying that
\begin{align*}
c(x|x) > \sum_{j \in P} x'_j c_j(x).
\end{align*}
\end{proof}

\subsection{$...$}

In this section, we show that the distance of the iterates of the discrete-time replicator dynamic as measured by the relative entropy with respect to the Wardrop equilibrium monotonically decreases, which implies that the relative entropy with respect to the Wardrop equilibrium is a Lyapunov function for the system. More precisely, given probability distributions $p$ and $q$ of a discrete random variable, their relative entropy is defined as $RE(p,q) = \sum_{j} p_j \ln \left( p_j / q_j \right)$. Note that $RE(p,q) \geq 0$ and $RE(p,q) = 0$ {\em iff} $p=q$. Letting $\zeta$ be the Wardrop equilibrium the Lyapunov function is $RE(\zeta,x)$.\footnote{Note that the terms in $RE(\zeta,x) = \sum_{i=1}^n \sum_{j \in P_i} \zeta^i_j \ln \left( \zeta^i_j / x_j \right)$ corresponding to $\zeta^i_j = 0$ are defined to be zero.} It remains to show that if $x$ is in the relative interior of $X$ and $x \mapsto \hat{x}$, then $RE(\zeta, \hat{x}) < RE(\zeta, x)$. 

\begin{theorem}
\label{wpofjbgkjbgxkjfg}
Let $\zeta$ be the Wardrop equilibrium and, for any point $x$ in the relative interior of $X$, let $x \mapsto \hat{x}$. Then $$RE(\zeta, \hat{x}) < RE(\zeta, x).$$
\end{theorem}

In the proof of Theorem~\ref{wpofjbgkjbgxkjfg}, we are going to need the following proposition (shown independently by~\cite{SelfishRoutingEvolution} and~\cite{StableGames}).

\begin{proposition}
\label{aslkdjfhsdklfjhd}
Let $\zeta$ be the Wardrop equilibrium and $x \in X$ where $x \neq \zeta$. Then
\begin{align*}
\sum_{i=1}^n \sum_{j \in P_i} \zeta^i_j c^i_j(x) < \sum_{i=1}^n \sum_{j \in P_i} x^i_j c^i_j(x)
\end{align*}
\end{proposition}

We are also going to need the following two lemmas.

\begin{lemma}
\label{xzcmvbnzxlkbv}
If $\zeta$ is the Wardrop equilibrium, then, for any $x \neq \zeta$ in $X$,
\begin{align*}
\sum_{i=1}^n \sum_{j \in P_i} \zeta^i_j g^i_j(x) > 1.
\end{align*}
\end{lemma}

\begin{proof}
We have
\begin{align*}
\sum_{i=1}^n \sum_{j \in P_i} z^i_j g^i_j(x) &= \sum_{i=1}^n \sum_{j \in P_i} z^i_j \frac{1 - \alpha c^i_j(x)}{1 - (\alpha / \omega_i) x^i \cdot c^i(x)}\\
&= \sum_{i=1}^n \frac{\omega_i - \alpha \sum_{j \in P_i} z^i_j c^i_j(x)}{1 - (\alpha / \omega_i) x^i \cdot c^i(x)}\\
&= \sum_{i=1}^n \omega_i \frac{1 - (\alpha / \omega_i) z^i \cdot c^i(x)}{1 - (\alpha / \omega_i) x^i \cdot c^i(x)}.
\end{align*}
Now, for any $0 \leq x < 1$, we have $1 / (1-x) = 1 + x + x^2 + \cdots$, and, therefore,
\begin{align*}
\sum_{i=1}^n \sum_{j \in P_i} z^i_j g^i_j(x) = \sum_{i=1}^n \omega_i \left(1 - (\alpha / \omega_i) z^i \cdot c^i(x) \right) \left( 1 + (\alpha / \omega_i) x^i \cdot c^i(x) + \left( (\alpha / \omega_i) x^i \cdot c^i(x) \right)^2 + \cdots \right)
\end{align*}
The right hand side can be written as
\begin{align*}
\sum_{i=1}^n \omega_i \left( 1 - (\alpha / \omega_i) z^i \cdot c^i(x) +(\alpha / \omega_i) x^i \cdot c^i(x)\right) + \\
\sum_{i=1}^n \left\{ \left( (\alpha / \omega_i) x^i \cdot c^i(x) \right)^2 - (\alpha / \omega_i)^2 (x^i \cdot c^i(x))(z^i \cdot c^i(x)) \right\} + \\
\sum_{i=1}^n \left\{ \left( (\alpha / \omega_i) x^i \cdot c^i(x) \right)^3 - (\alpha / \omega_i)^3 (x^i \cdot c^i(x))^2(z^i \cdot c^i(x)) \right\} +\\
\cdots
\end{align*}
Note now that, by Proposition~\ref{aslkdjfhsdklfjhd}, the first term is greater than one whereas regarding the $k$th term (k=2,3,\ldots) we have
\begin{align*}
\sum_{i=1}^n &\left\{ \left( (\alpha / \omega_i) x^i \cdot c^i(x) \right)^k - (\alpha / \omega_i)^k (x^i \cdot c^i(x))^{k-1}(z^i \cdot c^i(x)) \right\} = \\
&= \sum_{i=1}^n \left( (\alpha / \omega_i) x^i \cdot c^i(x) \right)^{k-1} \left( x^i \cdot c^i(x) - z^i \cdot c^i(x) \right)\\
&> \sum_{i=1}^n \min_\ell \left\{ \left( (\alpha / \omega_\ell) x^\ell \cdot c^\ell(x) \right)^{k-1} \right\} \left( x^i \cdot c^i(x) - z^i \cdot c^i(x) \right)\\
&= \min_\ell \left\{ \left( (\alpha / \omega_\ell) x^\ell \cdot c^\ell(x) \right)^{k-1} \right\} \sum_{i=1}^n \left( x^i \cdot c^i(x) - z^i \cdot c^i(x) \right)\\
&> 0
\end{align*}
where the last inequality follows by Proposition~\ref{aslkdjfhsdklfjhd}. This completes the proof.
\end{proof}

\begin{lemma}
\label{zlkjbgvsjbdsjgb}
\begin{itemize}
\item[(i)] Let $w,y,z \in X$. If
\begin{align*}
\sum_{i=1}^n \sum_{j \in P_i} w^i_j g^i_j(y) < \sum_{i=1}^n \sum_{j \in P_i} w^i_j g^i_j(z),
\end{align*}
then
\begin{align*}
\sum_{i=1}^n \sum_{j \in P_i} w^i_j \frac{1}{g^i_j(y)} > \sum_{i=1}^n \sum_{j \in P_i} w^i_j \frac{1}{g^i_j(z)}.
\end{align*}
\item[(ii)] In particular, if $\sum_{i=1}^n \sum_{j \in P_i} \zeta^i_j g^i_j(x) >1$, then $\sum_{i=1}^n \sum_{j \in P_i} \zeta^i_j (1 / g^i_j(x)) < 1$.
\end{itemize}
\end{lemma}

\begin{proof}
{\em (i)} Let $\bar{c}^i(y) = y^i \cdot c^i(y)$ and $\bar{c}^i(z) = z^i \cdot c^i(z)$. We have
\begin{align*}
\sum_{i=1}^n \sum_{j \in P_i} w^i_j \frac{1}{g^i_j(y)} &- \sum_{i=1}^n \sum_{j \in P_i} w^i_j \frac{1}{g^i_j(z)} =\\
&= \sum_{i=1}^n \sum_{j \in P_i} w^i_j \left( \frac{1}{g^i_j(y)} - \frac{1}{g^i_j(z)} \right)\\
&= \sum_{i=1}^n \sum_{j \in P_i} w_j^i \left( \frac{1 - (\alpha / \omega_i) \bar{c}^i(y)}{1 - \alpha c^i_j(y)} - \frac{1 - (\alpha / \omega_i) \bar{c}^i(z)}{1 - \alpha c^i_j(z)} \right)\\
&= \sum_{i=1}^n \sum_{j \in P_i} w_j^i \frac{(1 - (\alpha / \omega_i) \bar{c}^i(y))(1 - \alpha c^i_j(z)) - (1 - (\alpha / \omega_i) \bar{c}^i(z))(1 - \alpha c^i_j(y))}{(1 - \alpha c^i_j(y))(1 - \alpha c^i_j(z))}\\
&> \sum_{i=1}^n \sum_{j \in P_i} w_j^i \left((1 - (\alpha / \omega_i) \bar{c}^i(y))(1 - \alpha c^i_j(z)) - (1 - (\alpha / \omega_i) \bar{c}^i(z))(1 - \alpha c^i_j(y)) \right)\\
&> 0.
\end{align*}
The first inequality follows because $(1 - \alpha c^i_j(y))(1 - \alpha c^i_j(z)) < 1$ and the second inequality follows by the assumption of the lemma.\\

\noindent
{\em (ii)} Follows by the first part of the lemma if we set $w := \zeta$, $y := \zeta$, and $z := x$.
\end{proof}

\begin{proof}[Proof of Theorem~\ref{wpofjbgkjbgxkjfg}]
We have
\begin{align*}
RE(\zeta, \hat{x}) - RE(\zeta, x) &= \sum_{i=1}^n \sum_{j \in P_i} \zeta^i_j \ln \left( \frac{\zeta^i_j}{\hat{x}^i_j} \right) - \sum_{i=1}^n \sum_{j \in P_i} \zeta^i_j \ln \left( \frac{\zeta^i_j}{x_j} \right)\\
                       &= \sum_{i=1}^n \sum_{j \in P_i} \zeta^i_j  \ln \left( \frac{x^i_j}{\hat{x}^i_j} \right)\\
                       &= \sum_{i=1}^n \sum_{j \in P_i} \zeta^i_j  \ln \left( \frac{1- (\alpha / \omega_i) x^i \cdot c^i(x)}{1-\alpha c^i_j(x)} \right).
\end{align*}
Now, for any $y > 0$, $\ln y \leq y-1$, and, therefore,
\begin{align*}
RE(\zeta, \hat{x}) - RE(\zeta, x) &\leq \sum_{i=1}^n \sum_{j \in P_i} \zeta^i_j \left( \frac{1- (\alpha / \omega_i) x^i \cdot c^i(x)}{1-\alpha c^i_j(x)} - 1 \right)\\
                                         &= \sum_{i=1}^n \sum_{j \in P_i} \zeta^i_j \frac{1- (\alpha / \omega_i) x^i \cdot c^i(x)}{1-\alpha c^i_j(x)} - 1\\
                                         &= \sum_{i=1}^n \sum_{j \in P_i} \zeta^i_j \frac{1}{g^i_j(x)} - 1.
\end{align*}
By Lemma~\ref{xzcmvbnzxlkbv}, we know that $\sum_{i=1}^n \sum_{j \in P_i} \zeta^i_j \frac{1}{g^i_j(x)} > 1$, and, by Lemma~\ref{zlkjbgvsjbdsjgb}, this implies that $$\sum_{i=1}^n \sum_{j \in P_i} \zeta^i_j \frac{1}{g^i_j(x)} < 1,$$ which completes the proof.
\end{proof}

\section{Duality}

We have
\begin{align*}
\nabla (\nabla \cdot c) = \nabla \times (\nabla \times c) + \nabla^2 c.
\end{align*}
Therefore, if $c$ is a gradient field, we have
\begin{align*}
\nabla (\nabla \cdot \nabla f) &= \nabla \times (\nabla \times \nabla f) + \nabla^2 \nabla f\\
\nabla (\nabla \cdot \nabla f) &= \nabla^2 \nabla f\\
\nabla (\nabla^2 f) &= \nabla^2 \nabla f.
\end{align*}
This says that the gradient of the Laplacian is equal to the vector Laplacian of the gradient.

\bibliographystyle{plain}
\bibliography{axiomatic}

\fi

\end{document}